\documentclass[11pt, a4paper]{article}
\usepackage[margin=3cm]{geometry}

\usepackage{prem}
\usepackage{braket}
\usepackage{csquotes}
\usepackage{pst-node}
\usepackage{centernot}
\usepackage[affil-it]{authblk}

\date{\today}
\title{A mathematical foundation for self-testing: Lifting common assumptions}

\author[1]{Pedro Baptista}
\author[2]{Ranyiliu Chen}
\author[3]{J\k{e}drzej Kaniewski}
\author[4]{David Rasmussen Lolck}
\author[5]{Laura Man\v{c}inska}
\author[6]{Thor Gabelgaard Nielsen}
\author[7]{Simon Schmidt}

\affil[1]{Department of Computer Science, Federal University of Minas Gerais, Ave. Antônio Carlos, 6627, Belo Horizonte, MG, Brazil}
\affil[1,2,4,5,6,7]{Department of Mathematical Sciences, University of Copenhagen, Universitetsparken 5, 2100 Copenhagen \O, Denmark}
\affil[3]{Faculty of Physics, University of Warsaw, Pasteura 5, 02-093 Warsaw, Poland}
\affil[7]{Faculty of Computer Science, Ruhr University Bochum, Universitätsstra{\ss}e 150, 44801 Bochum, Germany}

\begin{document}
	\maketitle

\begin{abstract} 
     In this work we study the phenomenon of self-testing from the first principles, aiming to place this versatile concept on a rigorous mathematical footing. Self-testing allows a classical verifier to infer a quantum mechanical description of untrusted quantum devices that she interacts with in a black-box manner. Somewhat contrary to the black-box paradigm, existing self-testing results tend to presuppose conditions that constrain the operation of the untrusted devices. A common assumption is that these devices perform a \emph{projective} measurement of a \emph{pure} quantum state. Naturally, in the absence of any prior knowledge it would be appropriate to model these devices as measuring a mixed state using POVM measurements, since the purifying/dilating spaces could be held by the environment or an adversary. 
     
     We prove a general theorem allowing to remove these assumptions, thereby promoting most existing self-testing results to their assumption-free variants. On the other hand, we pin-point situations where assumptions cannot be lifted without loss of generality. As a key (counter)example we identify a quantum correlation which is a self-test only if certain assumptions are made. Remarkably, this is also the first example of a correlation that cannot be implemented using projective measurements on a bipartite state of full Schmidt rank. Finally, we compare existing self-testing definitions, establishing many equivalences as well as identifying subtle differences.
\end{abstract}
\newpage

\tableofcontents
\newpage
    \section{Introduction}

\emph{Self-testing} was first introduced by Mayers and Yao \cite{MY}. Over the years it has evolved into an active research field with widespread applications that well surpass the initial expectations (see review paper \cite{SB}). It could be argued that the successful applications of self-testing have outpaced a thorough examination and rigorous development of the underlying mathematical formalism. Addressing this issue is the overarching goal of this work which we believe will further broaden and ease future applications of self-testing. 

The original motivation of self-testing is that it can be used for certifying quantum devices. Since the properties of quantum systems are inherently difficult to observe directly, the problem of certifying that a quantum device functions according to its specification is challenging. Self-testing provides the strongest form of certification in this context. Specifically, it allows untrusted parties to convince a classical verifier that their shared quantum memory holds a specific state on which they are able to perform certain quantum measurements. 


In addition to its fundamental role in certification, self-testing techniques are increasingly being used for other purposes. These include protocols for delegated quantum computation, verifiable randomness generation, device-independent cryptography, Bell nonlocality, and quantum complexity theory. In fact some of the biggest recent breakthroughs like the $\mathrm{MIP}^*=\mathrm{RE}$ \cite{MIP=RE2020} crucially rest on self-testing techniques.

The concept of self-testing can be framed in the context of nonlocal games \cite{CHTW}, which involve two untrusted provers, Alice and Bob, and a verifier. The provers respond to questions from the verifier and their win or loss is determined by a predefined function. Crucially, Alice and Bob cannot communicate after receiving the questions but can agree on a strategy beforehand. In a quantum strategy, $S=(\rho,\{A_{xa}\},\{B_{yb}\})$, they share an entangled state, $\rho$, and employ local measurements to obtain their answers. We refer to a quantum strategy the provers use as an \textbf{arbitrary strategy} whereas a quantum strategy we would like to certify as a \textbf{canonical strategy}. In this context, self-testing posits that any arbitrary strategy that optimally wins a nonlocal game must be equivalent to a canonical strategy for that game, up to a local isometry. 
This conceptual idea of self-testing has 
been formalized in many different, albeit similar,
definitions 
\cite{MY, SB, GeoQCor,paddock2023operatoralgebraic} which have then been employed to establish self-testing theorems. This naturally evokes the following questions:

\begin{question}
\label{q:1}
   What is the relationship between the existing definitions of self-testing (and hence the obtained self-testing theorems)? Which is the strongest or the ``right'' definition of self-testing?
\end{question}

Understanding the difference between the existing definitions of 
self-testing is further complicated by the fact that most authors do not allow the arbitrary strategy to take the most general form allowed by quantum mechanics (POVM measurements on a mixed quantum state). A priori this weakens the resulting notion of self-testing and goes against the idea that the untrusted provers should be allowed to be all-powerful. Most of the existing self-testing results place at least one of the following three assumptions on the arbitrary strategy, $S=(\rho,\{A_{xa}\},\{B_{yb}\})$, employed by Alice's and Bob's untrusted quantum devices:
\begin{itemize}
    \item[(1)] the shared state, $\rho$, is pure,
    \item[(2)] the shared state has full Schmidt rank\footnote{For example, the state $\ket{\psi}=1/\sqrt{2}(\ket{00}+\ket{11})\in\mathbb{C}^2\otimes\mathbb{C}^2$ is full-rank, while $\ket{\psi'}=1/\sqrt{2}(\ket{00}+\ket{22})\in\mathbb{C}^3\otimes\mathbb{C}^3$ is not.},
    \item[(3)] the measurements $\{A_{xa}\}$ and $\{B_{yb}\}$ are projective measurements (PVM's). 
\end{itemize}

Depending on which assumptions are made on the arbitrary strategy, we refer to the resulting self-tests as pure/ full-rank/ PVM self-tests. For example, a PVM self-test means that any PVM strategy can be mapped by a local isometry to the canonical strategy but such a mapping need not exist for non-projective strategies\footnote{One might attempt to get rid of this assumption by dilating the non-projective strategy to a projective one with Naimark dilation. Unfortunately, the isometry that exists for Naimark dilation does not directly work for the original non-projective measurement in the sense of the self-testing definition. At a conceptual level this argument is problematic if we consider that the dilating space could be held by an adversary or the environment.}. If the arbitrary strategy is not restricted in any way, we say that the resulting self-test is \emph{assumption-free}. 
The above assumptions (1)--(3) give rise to a hierarchy of self-tests: for instance, every pure PVM self-test is also a PVM self-test. Hardly any of the existing self-tests are proven to be assumption-free. In particular, it is very common to only consider arbitrary strategies which measure a pure state with projective measurements (Assumptions (1) \& (3)). This leads to several intriguing questions:
\begin{question}
\label{q:2}
    Which of the assumptions (1)--(3) can (or cannot) be lifted without loss of generality? Can things go wrong in the most general case?
\end{question}
To gain intuition of the potential consequences of making unjustified assumptions, consider an example from \cite{Christandl2022operational} where two provers receive a single question each and produce a perfectly correlated bit. This can be achieved with a classical, separable mixed state: no quantum entanglement needed. However, if we assume that the perfectly correlated bit is produced by measuring a pure state, then this state needs to be entangled, leading to an entirely different analysis and conclusions. To give a more practical example, in device-independent random number generation, randomness is secure if it is not predictable by a third party \cite{Ac_n_2016}. Then the purity assumption oversimplifies and invalidate the security analysis, as there is no way any third party is entangled with a pure state. {The assumption that all measurements are projective is sometimes made for the sake of simplicity or due to historical precedent. On the other hand, we know that non-projective measurements are essential for certain tasks in quantum error correction and state discrimination. Adhering to this assumption could therefore unnecessarily restrict the applicability of self-testing methods.} From a philosophical standpoint, making additional assumption goes against the idea of self-testing, which aims to make as few assumptions as possible. This is particularly important in cryptographic contexts where fewer assumptions often translate into stronger security guarantees.

Shifting the attention to the canonical strategy, we can inquire about the limits of self-testing, namely, which states and measurements can we hope to self-test: 
\begin{question}
\label{q:3}
    Which strategies can or cannot be self-tested?
\end{question}
It is known that mixed states cannot be self-tested (see {\it e.g.} \cite[Sect. 3.5]{SB}). 
However, it remains uncertain whether strategies containing non-full-rank pure states or non-projective measurements can be self-tested. Additionally, the ability to self-test certain strategies might depend on the specific definition of self-testing or the assumptions regarding the arbitrary strategies employed. 

Moving beyond self-testing, the 2-party Bell scenario presents an intriguing question: what is the \emph{simplest} form of a strategy realizing a given bipartite correlation? Conventional purification and Naimark dilation arguments show that any quantum correlation can be realized by measuring a pure state with local projective measurements. In a similar vein, by restricting a strategy to the state's local supports, we can obtain a strategy which utilizes a pure state of full Schmidt rank. However, these standard arguments fall short of providing a strategy that simultaneously has all three desired attributes: purity, full rank, and projective measurements. This leads us to the following question: 
\begin{question}
    \label{q:4}
    Can every bipartite quantum correlation be realized by locally measuring a shared state of full Schmidt rank with projective measurements?
\end{question}

\subsection{Results}

We study self-testing from first principles. This includes determining the assumptions that can be made without sacrificing generality, as well as identifying which quantum states and measurements are amenable to self-testing. At a conceptual level, we make the following contributions:

\begin{itemize}
    \item We put the concept of self-testing on a rigorous mathematical footing. {This includes identifying new key concepts ({\it e.g.} ``support-preserving strategy''), putting forth definitions for robust versions of properties like ``projectivity'' or ``being support-preserving'', and unifying the existing self-testing definitions.};
    \item We establish which strategies can or cannot be self-tested in an assumption-free manner,
    \item We delineate which combinations of assumptions (1)--(3) can or cannot be lifted;
    \item We provide a previously unknown reference example of a quantum correlation that cannot be realized by projective measurements on a full-rank state.
\end{itemize}

To start things off, we examine and compare existing self-testing definitions. We show that some of these definitions are indeed equivalent while highlighting subtle differences and pin-pointing assumptions which potentially cannot be lifted, thus giving a better understanding of Question \ref{q:1}.

\begin{theoremA}(informal version of the results in Section \ref{sect:equivalence})
   The existing definitions of self-testing are equivalent in certain natural settings. In the general case we identify definitions that yield weaker notion of self-testing (\textit{See also Theorem \ref{thmB2}}). 
\end{theoremA}

To address Question \ref{q:2} we establish the following theorem:

\begin{theoremB1}\makeatletter\def\@currentlabel{B.1}\makeatother\label{thmB1}(Theorem \ref{thm:bigtheorem})~
Let $G$ be a nonlocal game.
\begin{enumerate}
    \item[(a)] Let $\tilde{S}$ be a strategy for $G$ that measures a pure state of full Schmidt rank.\\
    If $G$ is a pure PVM self-test of $\tilde{S}$  then it is also an assumption-free self-test of  $\tilde{S}$.
    \item[(b)] Let $\tilde{S}$ be a strategy for $G$ that uses only projective measurements.\\
    If $G$ is a pure full-rank self-test of $\tilde{S}$  then it is also an assumption-free self-test of~$\tilde{S}$.
\end{enumerate}
\end{theoremB1}

\begin{remark*}
Theorem~\ref{thmB1} (as well as Theorem \ref{thmC} given below) also hold for self-tests from Bell inequalities and extreme quantum correlations. In Section \ref{sec:lifting} we establish robust versions of these results.
\end{remark*}

Theorem \ref{thmB1} shows that both the ``purity + projectivity'' and ``purity + full-rank'' assumptions can consistently be lifted, thus {\it elevating} most (if not all) of the existing self-tests to their assumption-free versions. We view this as the most practically impactful contribution of our work. Essentially, Theorem \ref{thmB1} enables us to sidestep cumbersome general strategies that involve non-projective measurements and arbitrary mixed states, thereby simplifying the proof process for self-testing theorems. Part (a) of Theorem \ref{thmB1} answers the question raised in \cite[Appendix B.2]{SB}.

It is natural to ask whether without loss of generality we can restrict to full-rank \emph{and} projective arbitrary strategies. We conjecture that this is not possible in general, since we do not know a general construction allowing to promote an arbitrary strategy to an equivalent strategy that is simultaneously full-rank and projective. 

\begin{conjecture}
  There is a nonlocal game that is a pure full-rank PVM self-test for a projective full-rank strategy $\tilde S$, which is not an assumption-free self-test for $\tilde S$. 
\end{conjecture}

Figure \ref{fig:diagram} succinctly illustrates Theorem \ref{thmB1} as well as our conjecture.

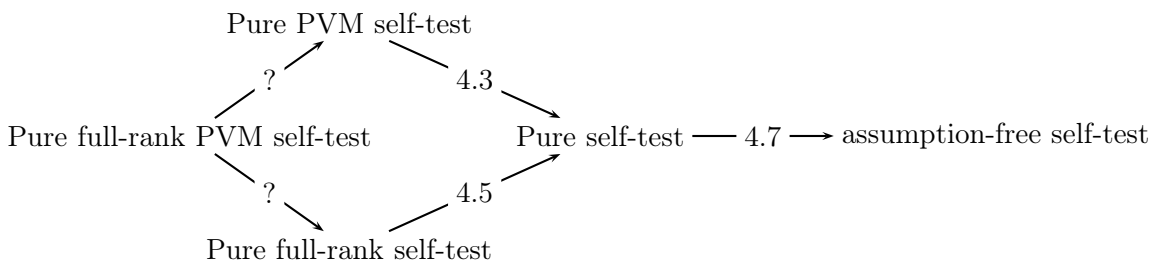
\begin{figure}[ht]
\begin{pspicture}(16,3.7)


\rput[c](2.4,2){\rnode{00}{Pure full-rank PVM self-test}}

\rput[c](4.5,0.5){\rnode{01}{Pure full-rank self-test}}

\rput[c](4.5,3.5){\rnode{02}{Pure PVM self-test}}

\rput[c](7.8,2){\rnode{A}{Pure self-test}}
\rput[c](13,2){\rnode{B}{assumption-free self-test}}

\ncline[nodesep=3pt]{->}{00}{01}
\ncput*{?}
\ncline[nodesep=3pt]{->}{00}{02}
\ncput*{?}
\ncline[nodesep=3pt]{->}{A}{B}
\ncput*{\ref{thm:pure_to_mix_rob}}
\ncline[nodesep=3pt]{->}{01}{A}
\ncput*{\ref{thm:full_to_any_rank_rob}}
\ncline[nodesep=3pt]{->}{02}{A}
\ncput*{\ref{thm:PVMtoPOVMrob}}
\end{pspicture}
\caption{Implications of Theorem \ref{thmB1} for a  canonical full-rank projective strategy. Specifically, Theorem \ref{thmB1} (a) consists of Theorem \ref{thm:PVMtoPOVMrob} and \ref{thm:pure_to_mix_rob}. And Theorem \ref{thmB1} (b) consists of Theorem \ref{thm:full_to_any_rank_rob} and \ref{thm:pure_to_mix_rob}. For the implications with `?' we conjecture them to be false.}
\label{fig:diagram}
\end{figure}

Taking another look at Theorem \ref{thmB1}, one can ask whether the same conclusions hold if the canonical strategy $\tilde S$ is not projective and/or full-rank.  We show that in any more general situation, there are examples of restricted self-tests that are not assumption-free self-tests, giving an answer to the second part of Question \ref{q:2}. 

\begin{theoremB2}\makeatletter\def\@currentlabel{B.2}\makeatother\label{thmB2}(Corollary \ref{cor:fullrankselftestnoselftest} and Corollary \ref{cor:PVMselftestnoselftest})
We exhibit an extreme quantum correlation $\tilde{p}$ that can be realized by two different strategies. The first strategy, $S_\textrm{full-rank}$, measures a full-rank shared state with non-projective measurements while the second strategy, $S_\textrm{proj}$, measures a non-full-rank shared state with projective measurements\footnote{Essentially we can take ${S}_\textrm{proj}$ to be the Naimark dilation of ${S}_\textrm{full-rank}$.}. We have that
\begin{enumerate}
    \item[(a)] $\tilde p$ is a PVM self-test of $S_\textrm{proj}$,
    \item[(b)] $\tilde p$ is a full-rank self-test of $S_\textrm{full-rank}$,
    \item[(c)] $\tilde p$ is {\bf not} an assumption-free self-test (no matter what canonical strategy we choose),
    \item[(d)] $\tilde p$ cannot be realized by performing projective measurements on a shared state with full Schmidt rank. (This answers Question \ref{q:4} in the negative.)
\end{enumerate}
\end{theoremB2}

The correlation $\tilde p$ highlights that we need to be very careful with the assumptions imposed on the arbitrary strategies as these assumptions can yield weaker versions of self-testing. Indeed, assuming that the arbitrary strategy is projective (or full-rank) allows us to obtain a self-testing result for $\tilde p$ while this fails to hold in the absence of any assumptions.

To obtain the correlation $\tilde p$ we combine the CHSH inequality with another Bell inequality, in which Alice holds the same measurement operators as in CHSH and Bob gets an additional three-outcome measurement. 

Finally, one might ask if it is possible to self-test canonical strategies $\tilde S$ that are not full-rank or projective.
Using the properties of Naimark dilation and restriction, we show that this is essentially not the case, thus answering Question \ref{q:3}. 

\begin{theoremC}\makeatletter\def\@currentlabel{C}\makeatother\label{thmC} (An abridged version of Theorem \ref{thm:originalpartC})
If a nonlocal game is an assumption-free self-test, then there exists a full-rank projective strategy that is self-tested by this game. Moreover, full-rank non-projective strategies cannot be self-tested in an assumption-free way.
\end{theoremC}
A crucial take-away from the above Theorem~\ref{thmC} is that it is not possible to self-test non-projective measurements.

We remark that for the case of pure self-tests from correlations, a conclusion similar to the first part of Theorem~\ref{thmC} has also been reached in \cite[Proposition 4.14]{paddock2023operatoralgebraic} via a different approach.

\subsection{Structure of the paper}
In Section \ref{sec:prem}, we recall basic notions in nonlocal games and explain the different self-testing definitions we are using throughout the paper. Section \ref{sec:lemmas} gives the tools we are using for the main results. We show how strategies being support-preserving or projective is connected to local dilation. In addition, we look at restrictions and Naimark dilations of strategies which can also be related to such strategies. For an overview of this, see Figure \ref{fig:visualizingConcepts}. Then, in Section \ref{sec:lifting}, we have all the tools to prove our main result: Theorem \ref{thmB1}. Here we lift common assumptions in self-testing theorems, in different steps. First, we lift the PVM assumption, then the full-rank assumption and finally the assumption that the state in an arbitrary strategy is pure. Using the tools from the previous section, we are also able to prove Theorem \ref{thmC} here. After this, we compare existing definitions for local dilations for strategies using mixed states and show their equivalence in Section  \ref{sect:equivalence}. Finally, in Section \ref{sect:counterexamples}, we proceed by proving Theorem \ref{thmB2}. Here, we give examples showing we cannot lift our assumptions on self-testing theorems in general. 
	\section{Preliminaries}\label{sec:prem}
\subsection{Notation}
We state some notations and basic facts we need throughout the article. Unless specified otherwise, we assume all Hilbert spaces to be finite-dimensional. For any vector $v\in\mathcal{H}$, we denote its norm by $\|v\|=\braket{v,v}^{\frac{1}{2}}$. {We write $u \approx_{\varepsilon}v$ if $\|u-v\|\le\varepsilon$.}

A \emph{pure state} $\ket{\psi}$ is a unit vector in a Hilbert space $\mathcal H$. For a bipartite pure state $\ket\psi \in \mathcal H_A \otimes \mathcal H_B$, we can consider its \emph{Schmidt decomposition}
\begin{equation}
\ket\psi = \sum_{i=0}^{k-1}\alpha_i \ket{e_i}\ket{f_i}
\label{eq:Sch}
\end{equation}
where $\alpha_i > 0$ and both $\left\{\ket{e_i}\right\}_{i=0}^{k-1} \subseteq \mathcal H_A$ and $\left\{\ket{f_i}\right\}_{i=0}^{k-1} \subseteq \mathcal H_B$ are orthonormal sets. Note that \pref{eq:Sch} only includes terms with positive Schmidt coefficients. We refer to the number $k$ as the \emph{Schmidt rank} of the state $\ket\psi$ and if $k = \operatorname{dim}(\mathcal H_A) = \operatorname{dim}(\mathcal H_B)$ then we say that $\ket\psi$ has full Schmidt rank, or $\ket\psi$ is \emph{full-rank} for simplicity.  We also define $\Supp_A \ket\psi := {\Span\left\{\ket{e_0}, \dotsc \ket{e_{k-1}} \right\}} \subseteq \mathcal H_A$, and similarly, $\Supp_B \ket\psi := {\Span\left\{\ket{f_0}, \dotsc, \ket{f_{k-1}} \right\}} \subseteq \mathcal H_B$.

A \emph{mixed state} $\rho$ is represented by a positive semi-definite, self-adjoint matrix with trace one, also called density matrix. Every mixed state can be written in the form $\rho=\sum_{i=1}^n p_i \proj{
\psi_i}$ for a probability vector $(p_i)_{i=1}^n$ and pure states $\ket{\psi_i}$. A \emph{purification} of a mixed state $\rho\in B(\mathcal H)$ is a pure state $\ket{\psi} \in \mathcal H \otimes \mathcal H_P$ for some purification Hilbert space $\mathcal H_P$ such that $\rho =\tr_P(\proj{\psi})$. Here $\mathrm{Tr}_P$ denotes the partial trace over the Hilbert space $\mathcal H_P$.

We will be working with the following definition of measurements.
\begin{definition}[POVM]
    A \emph{positive, operator-valued measurement (POVM)} is a set of positive, self-adjoint operators $\{E_i\}_{i=1}^n$ in $B(\mathcal H)$ such that 
    $$\sum_{i=1}^n E_i = \mathbbm{1}_{B(\mathcal H)}.$$
    Furthermore, if all operators are projections ($E_i=E_i^2=E_i^*$ for all $i$), then we call it \emph{projective measurement (PVM)}.
\end{definition}

\subsection{Nonlocal games and strategies}
\begin{definition}[Nonlocal game]\label{def:bck_nonlocalgame}
	A nonlocal game $G$ is a tuple $(\mathcal S,\mathcal T,\mathcal A,\mathcal B, \pi, \mathcal V)$ of a probability distribution of the questions $\pi: \mathcal{S}\times \mathcal{T}\to [0,1]$ and a verification function $\mathcal{V}: \mathcal{A}\times\mathcal{B}\times\mathcal{S}\times\mathcal{T} \to \{0,1\}$,
	where $\mathcal{S}$ and $\mathcal{T}$ are finite sets of questions for Alice and Bob respectively, and  $\mathcal{A}$ and $\mathcal{B}$ are finite sets of answers for Alice and Bob respectively.
\end{definition}

A nonlocal game $G$ is a cooperative game played by two players, which we typically call Alice and Bob, and a referee. The game is played the following way: Before the game starts, Alice and Bob agree on some strategy, including possibly sharing a quantum state. During the game, the players are not allowed to communicate, but they can perform measurements on their own parts of the shared state. Alice and Bob each receive a question $s$ and $t$ respectively from predetermined sets of questions $\mathcal{S}$ and $\mathcal{T}$, determined by the probability distribution $\pi$. Each of them then gives an answer $a$ and $b$ from their answer sets $\mathcal A$ and $\mathcal B$, respectively. They then win if $\mathcal{V}(a,b|s,t) = 1$. In such games, the behaviour of Alice and Bob are described as \emph{quantum strategies}.

\begin{definition}[Strategy]
	A (tensor-product) quantum strategy for a nonlocal game $G = (\mathcal S,\mathcal T,\mathcal A,\mathcal B, \pi, \mathcal V)$ is a tuple
	\begin{equation}
	    \label{eq:general-strategy-definition}
	    S=(\rho_{AB},\{A_{sa}\}_{s\in\mathcal{S},a\in\mathcal{A}},\{B_{tb}\}_{t\in\mathcal{T},b\in\mathcal{B}}),
	\end{equation}
	consisting of a shared density operator $\rho_{AB}\in B(\mathcal{H}_A\otimes\mathcal{H}_B)$, where $\mathcal{H}_A$ is the state space of Alice and $\mathcal{H}_B$ is the state space of Bob. Furthermore, for each $s\in\mathcal{S}$, the set $\{A_{sa}\}_{a\in\mathcal{A}}\subset B(\mathcal{H}_A)$ is a POVM on $\mathcal{H}_A$, and for each $t\in \mathcal{T}$, the set $\{B_{tb}\}_{b\in\mathcal{B}}\subset B(\mathcal{H}_B)$ is a POVM on $\mathcal{H}_B$.
	We identify the following special cases (which are not mutually exclusive):
	\begin{itemize}
	    \item If $\rho_{AB} = \proj\psi$ for some pure state $\ket\psi \in \mathcal H_A \tensor \mathcal H_B$, we refer to the quantum strategy as \emph{pure}. In this case, we may replace $\rho_{AB}$ with $\ket\psi$ in \eqref{eq:general-strategy-definition}. 
     \item If both marginal states $\rho_{A}:=\tr_B[\rho_{AB}],\rho_{B}:=\tr_A[\rho_{AB}]$ have rank equal to the dimension of corresponding Hilbert space, we may refer to the quantum strategy as \emph{full-rank}. In the case of pure state $\rho_{AB}=\ket{\psi}$, this is equivalent to $\ket{\psi}$ having full Schmidt rank.
	    \item If all POVM elements $A_{sa}$ and $B_{tb}$ are projectors, then we refer to the quantum strategy as \emph{projective}. Otherwise, we call it \emph{non-projective}.
	\end{itemize}
\end{definition}

We will write $\{A_{sa}\}_{s\in\mathcal{S},a\in\mathcal{A}}$ as $\{A_{sa}\}$ when from the context it is clear that the set is indexed over the sets $\mathcal{S}$ and $\mathcal{A}$. We will use analogous notation for Bob's measurements $\{B_{tb}\}$. We will refer to a quantum strategy simply as a strategy in the following.

It is easy to compute the winning probability when using a particular strategy $S$ for a game $G$. Sometimes, it will be useful to collect all the information regarding game $G$ and the employed measurements in a single operator $W$.
\begin{lemma}\label{lem:bck_score}
	Let $G = (\mathcal S,\mathcal T,\mathcal A,\mathcal B, \pi, \mathcal V)$ be a nonlocal game and $S=(\rho_{AB},\{A_{sa}\},\{B_{tb}\})$ a strategy for $G$. Define $W$ as 
	\begin{equation*}
		W := \sum_{a,b,s,t} \pi(s,t)\mathcal{V}(a,b|s,t)(A_{sa}\otimes B_{tb}).
	\end{equation*}
	Then the probability of winning the game $\omega(S,G)$ using the strategy $S$ can be found as
	\begin{equation*}
		\omega(S,G) = \tr(W\rho_{AB}).
	\end{equation*}
\end{lemma}
\begin{proof}
	We show this by rewriting $\omega(S,G)$ using the definition of $W$ and linearity of the trace,
	\begin{align*}
		\omega(S,G) &= \sum_{s,t}\pi(s,t)\sum_{a,b}\mathcal{V}(a,b|s,t) \tr((A_{st}\otimes B_{tb})\rho_{AB})\\
		&= \tr(\left(\sum_{s,t}\sum_{a,b}\pi(s,t)\mathcal{V}(a,b|s,t) (A_{st}\otimes B_{tb})\right)\rho_{AB})\\
		&= \tr(W\rho).
	\end{align*}
\end{proof}
Furthermore, we define $\omega_q(G) := \sup_S \omega(S, G)$, where the supremum is taken over all (tensor-product) strategies $S$ which are compatible with $G$. We refer to $\omega_q(G)$ as the \emph{optimal} or \emph{maximal} quantum value of $G$. We call $S$ an optimal strategy if $\omega(S,G)=\omega_q(G)$. Note that in general, it may not be possible to obtain $\omega_q(G)$ with any tensor-product strategy\cite{SlofstraQcor}. For  $\delta\geq 0$, we say that $S$ is $\delta$-optimal if $\omega(S,G) \geq \omega_q(G)-\delta$.

\subsection{Self-testing: Definitions}

The main topic of this work, \emph{self-testing}, asks whether the state and the measurements used in an optimal quantum strategy for a nonlocal game are unique, up to local isometries. The reason this is useful is that it enables us to make conclusions about the strategy only from the observed probability of winning. Consider a situation where we make Alice and Bob play a nonlocal game that self-tests an optimal strategy. If in this setting we observe that Alice and Bob achieve the quantum value of the game, we can guarantee that Alice and Bob must have used a strategy that is equivalent to the optimal one up to local isometries. We will also discuss the concept of \emph{robust self-testing}. Here, we ask that the states and measurements of an almost optimal strategy are close to those of a (canonical) optimal strategy, up to local isometries. 

In the literature, strategies are often assumed to use pure states and projective measurements, and so common formulations of self-testing reflect this fact. We will be working with a definition of self-testing that is very similar to the one presented by \cite[Definition 2]{SB}, though augmented with additional qualifiers that are relevant to the presentation of our results.

To start discussing the concept of (robust) self-testing, we first introduce the concept of local $\varepsilon$-dilations. 

\begin{definition}[Local $\varepsilon$-dilation]
Given two strategies
    \begin{align*}
        S&=(\rho_{AB}\in B({\mathcal H}_A \tensor {\mathcal H}_B),\{ A_{sa}\}_{s\in\mathcal{S},a\in\mathcal{A}},\{ B_{tb}\}_{t\in\mathcal{T},b\in\mathcal{B}}) 
        \text{ and}
        \\
        \tilde S&=(\ket{\tilde \psi} \in {\mathcal H}_{\tilde{A}} \tensor {\mathcal H}_{\tilde{B}},\{\tilde A_{sa}\}_{s\in\mathcal{S},a\in\mathcal{A}},\{\tilde B_{tb}\}_{t\in\mathcal{T},b\in\mathcal{B}})
    \end{align*}
we say that \emph{$\tilde S$ is a local $\varepsilon$-dilation of $S$} and write $S\xhookrightarrow{\varepsilon} \tilde S$ if for any purification $\ket{\psi} \in \mathcal{H}_A \tensor \mathcal{H}_B \tensor \mathcal{H}_P$ of $\rho_{AB}$ there exist spaces $\mathcal H_{\hat A},\mathcal H_{\hat B},$ a local isometry $U = U_A \tensor U_B$, with $U_A : \mathcal H_A \to \mathcal H_{\tilde A} \tensor \mathcal H_{\hat A}$, $U_B : \mathcal H_B \to \mathcal H_{\tilde B} \tensor \mathcal H_{\hat B}$ and a state $\ket{\aux} \in \mathcal H_{\hat A}\tensor\mathcal H_{\hat B}\tensor \mathcal H_{P}$ such that for all $s,t,a,b$ we have
    \begin{align}
             \|(U \tensor \mathbbm 1_P)\ket{\psi} - \ket{\tilde\psi} \tensor \ket{\aux}\|\leq \varepsilon,\nonumber\\
         \|(U \tensor \mathbbm 1_P) (A_{sa} \tensor \mathbbm{1}_B \tensor\mathbbm 1_P )\ket{\psi} - (\tilde A_{sa} \tensor \mathbbm 1_{\tilde B})\ket{\tilde\psi} \tensor \ket{\aux}\|\leq \varepsilon,\label{eq:localdilation}\\
         \|(U \tensor \mathbbm 1_P) (\mathbbm 1_A \tensor B_{tb} \tensor\mathbbm 1_P )\ket{\psi} - (\mathbbm 1_{\tilde A} \tensor \tilde B_{tb})\ket{\tilde\psi} \tensor \ket{\aux}\|\leq \varepsilon.\nonumber
    \end{align}  

In case we want to name the local isometry and the auxiliary state, we write $S \xhookrightarrow[U,\ket{\aux}]{\varepsilon}\tilde{S}$. We will use this notation only when $\rho_{AB}$ is pure to avoid ambiguity.
\label{def:localdilation}
\end{definition}

\begin{remark}~
    \begin{itemize}
        \item Note that local dilations are transitive. That is if $S_X\xhookrightarrow{\varepsilon_1} S_Y$ and $S_Y \xhookrightarrow{\varepsilon_2} S_Z$, then $S_X\xhookrightarrow{\varepsilon_1+\epsilon_2}S_Z$, see \cite[Lemma 4.7]{MPS}.
        \item If the state $\rho_{AB}=\proj{\psi}$ in strategy $S$ is pure, we do not need to concern ourselves with purifications of $\rho_{AB}$ in the above definition. That is, the auxiliary state $\ket{\aux}\in\mathcal{H}_{\hat{A}}\otimes \mathcal{H}_{\hat{B}}$, and Eq. \eqref{eq:localdilation} becomes 
        \begin{align*}
             \|U\ket{\psi} - \ket{\tilde\psi} \tensor \ket{\aux}\|\leq \varepsilon,\\
         \|U (A_{sa} \tensor \mathbbm{1}_B)\ket{\psi} - (\tilde A_{sa} \tensor \mathbbm 1_{\tilde B})\ket{\tilde\psi} \tensor \ket{\aux}\|\leq \varepsilon,\\
         \|U(\mathbbm 1_A \tensor B_{tb} )\ket{\psi} - (\mathbbm 1_{\tilde A} \tensor \tilde B_{tb})\ket{\tilde\psi} \tensor \ket{\aux}\|\leq \varepsilon.
    \end{align*}  
        \item  If $\varepsilon=0$ holds, we say that $\tilde S$ is a local dilation of $S$ and write $S \xhookrightarrow{} \tilde{S}$. For pure states, this is equivalent to finding a local isometry $U = U_A \tensor U_B$ such that
\[
U (A_{sa} \tensor B_{tb})\ket{\psi} = (\tilde A_{sa} \tensor \tilde B_{tb})\ket{\tilde\psi} \tensor \ket{\aux}
\]
holds for all $a,b,s,t$.
    \end{itemize}
\end{remark}

 Intuitively, self-testing allows us to say that \emph{any} optimal strategy  $S$ for a game $G$ can be mapped to a chosen canonical strategy $\tilde S$. In practice, however, when proving self-testing theorems, authors often impose different restrictions on the set of considered strategies $S$. Three most common types of assumptions restricting the strategy, $S$, implemented by the untrusted black-box quantum device are as follows:
 \begin{enumerate}
 \item the state in $S$ is pure rather than mixed,
 \item the state in $S$ is full-rank.
 \item the measurements in $S$ are projective rather than general POVMs,
 \end{enumerate}
 The above assumptions give rise to a  priori different definitions of self-testing. A \emph{$t$-strategy} for $t\subseteq {\{\text{pure},\text{full-rank},\text{PVM}\}}$
 is a strategy for the game, where the states and measurements are restricted according to~$t$. For example, a pure PVM strategy has a pure state and projective measurements, while the rank of the state can be arbitrary. An assumption-free strategy will usually just be called a strategy. 

\begin{definition}[Self-testing]\label{def:selftest}
   Let $\tilde{S}$ be a pure strategy and 
   $t\subseteq {\{\text{pure},\text{full-rank},\text{PVM}\}}$. We say that a nonlocal game $G$ is a \emph{$t$-self-test} for a (reference) strategy $\tilde{S}$ if  $S \xhookrightarrow{} \tilde{S}$ for every optimal $t$-strategy $S$ for game $G$.
\label{def:selftestDil}
\end{definition}

All those different definitions also have a robust version, defined as follows. 

\begin{definition}[Robust self-testing]\label{def:robustselftest}
    Let $\tilde{S}$ be a pure strategy, $t\subseteq {\{\text{pure},\text{full-rank},\text{PVM}\}}$. We say that a nonlocal game $G$ is a \emph{robust $t$-self-test} for a (reference) strategy $\tilde{S}$ if from every $\varepsilon \geq 0$, there exists $\delta\geq 0$ such that $S \xhookrightarrow{\varepsilon} \tilde{S}$ for every $\delta$-optimal $t$-strategy $S$ for a game $G$.
\end{definition}

Many of the existing robust self-testing results specify an explicit dependence between $\varepsilon$ and $\delta$. Our results also apply to this case, and we remark how the theorem statements should be altered in case one wishes to apply them to lift assumptions for a robust self-test with an explicit $(\varepsilon,\delta)$-dependence.

It is clear that every $t$-self-test is also a $t'$-self-test if $t'$ imposes more restrictions on the strategy than $t$. For example, every PVM self-test is also a pure PVM self-test.  We will refer to an assumption free self-test just as self-test. Note that in the literature, the term ``self-test'' is often used for a prior weaker form of self-testing, such as a pure PVM self-test. In this paper, we refer to an (assumption-free) self-test if our arbitrary strategies are allowed to have mixed states of any rank and POVM measurements. 
	\section{Key tools and concepts}\label{sec:lemmas}

In this section, we introduce support-preserving strategies as a fundamental concept crucial to the proofs of several key results. Furthermore, we examine the formalization of projectiveness, providing a comprehensive framework for the Naimark dilation of non-local strategies and presenting important properties associated with it. The interaction of the concepts introduced in this section can be visualized, see Fig. \ref{fig:visualizingConcepts}.

\begin{figure}[ht]
\begin{pspicture}(15,6.5)

\rput[c](2,6){\ovalnode{A}{$\varepsilon$-support-preserving}}
\rput[c](12,6){\ovalnode{B}{$\varepsilon$-projective}}

\rput[c](2,2){\ovalnode{D}{restriction}}
\rput[c](12,2){\ovalnode{E}{Naimark dilation}}

\rput[c](2,0.5){\rnode{F}{\psframebox{projective restriction \ref{thm:resProj}}}}
\rput[c](12,0.5){\rnode{G}{\psframebox{support-preserving Naimark \ref{thm:NaiSupp}}}}

\ncline[nodesep=3pt,offset=2pt]{->}{D}{A}
\ncline[nodesep=3pt,offset=2pt]{->}{A}{D}
\nbput[nrot=:D,labelsep=.25]{$\varepsilon$-dilation \ref{prop:restrictrob}}

\ncline[nodesep=3pt,offset=2pt]{->}{B}{E}
\ncline[nodesep=3pt,offset=2pt]{->}{E}{B}
\nbput[nrot=:D,labelsep=.25]{$\varepsilon$-dilation \ref{prop:Naimarkrob}}

\ncline[nodesep=3pt,linestyle=dashed]{->}{D}{F}
\ncline[nodesep=3pt,linestyle=dashed]{->}{B}{F}
\ncline[nodesep=3pt,linestyle=dashed]{->}{E}{G}
\ncline[nodesep=3pt,linestyle=dashed]{->}{A}{G}

\rput[c](7,4){\dianode[fillstyle=solid,fillcolor=white]{C}{local dilation}}

\ncline[nodesep=3pt]{-}{C}{A}
\nbput[nrot=:D]{invariant \ref{prop:invariantSupp}}

\ncline[nodesep=3pt]{-}{B}{C}
\nbput[nrot=:D]{invariant \ref{prop:invariantProj}}

\end{pspicture}
\caption{Local dilation ``$\xhookrightarrow{}$'' is the central concept of self-testing. We introduce the idea of support-preservingness and projectiveness for strategies, which are invariant under local dilations. There are two canonical ways of obtaining a support-preserving strategy and a projective strategy, respectively: restriction and Naimark dilation. If a strategy $S$ is support-preserving/projective, then we can locally dilate $S$ to its restriction/Naimark dilation and vice versa. Finally, if the restriction of $S$ is projective, or if its Naimark dilation is support-preserving, then $S$ must be both projective and support-preserving.}
\label{fig:visualizingConcepts}
\end{figure}
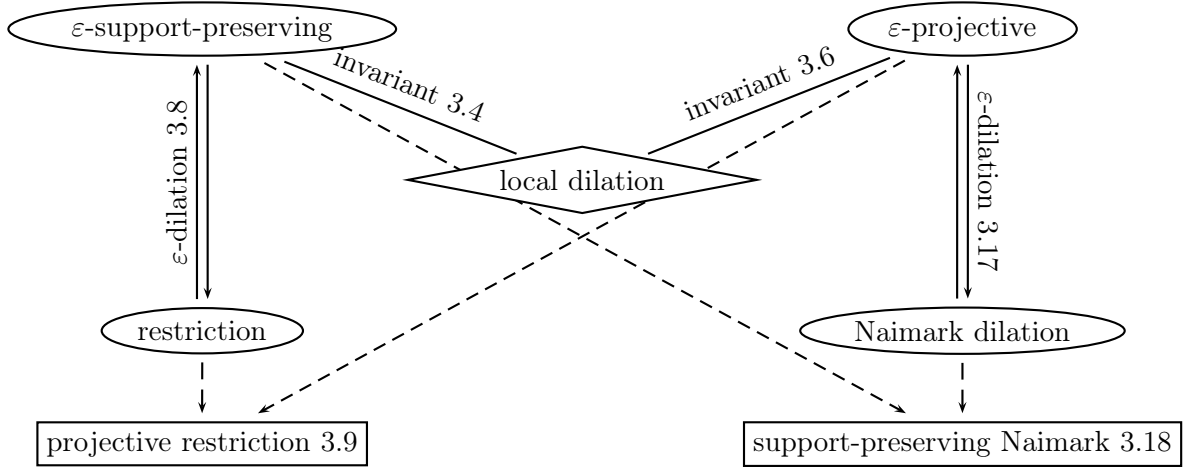

\subsection{Facts about local dilations}

We noted in Section 2 that the local dilation is transitive, so it is a pre-order on the set of strategies. In general, local dilation is not an equivalence relation, because if we let $S'$ to be $S$ attached with an {entangled} auxiliary state, then $S'\xhookrightarrow{}S$ but not the other direction. Nevertheless, we can show that if the auxiliary state in a local dilation is separable, then the two strategies are ``equivalent'' in the sense of local dilation:

\begin{proposition}
    If a strategy $S_n$ is a $\varepsilon$-local dilation of a strategy $S_m$ with a separable ancilla:
    $$
    S_m\xhookrightarrow[V_A\otimes V_B,\ket{0}_{A'}\otimes\ket{0}_{B'}]{\varepsilon}S_n,
    $$
    then $S_m$ is also a $\varepsilon$-local dilation of $S_n$ (for some separable ancilla).
    \label{prop:otherdirection}
\end{proposition}
\begin{proof}
    Without loss of generality, assume that $S_n$ has local dimension $n$, and $S_m$ has local dimension $m$. Since $S_m\xhookrightarrow[V_A\otimes V_B,\ket{0}_{A'}\otimes\ket{0}_{B'}]{\varepsilon}S_n$, then
    \begin{align*}
        (V_A\otimes V_B)(A_{sa}^{(m)}\otimes B_{tb}^{(m)})\ket{\psi_m}\approx_\varepsilon&(\ket{0}_{A'}\ket{0}_{B'})\otimes(A_{sa}^{(n)}\otimes B_{tb}^{(n)})\ket{\psi_n}\\
        =&(\mathbbm{1}_{n'_A\times n}\otimes \mathbbm{1}_{n'_B\times n})(A_{sa}^{(n)}\otimes B_{tb}^{(n)})\ket{\psi_n}
    \end{align*}
    where $n'_A=n\times\dim\mathcal{H}_{A'}$, $n'_B=n\times\dim\mathcal{H}_{B'}$, and $\mathbbm{1}_{x\times y}$ ($y\le x$) denotes the first $y$ columns of the $x\times x$ identity matrix, which is an isometry. 
    
    Express $V_A,V_B$ as $V_A=U_A\mathbbm{1}_{n'_A\times m},V_B=U_B\mathbbm{1}_{n'_B\times m}$, where $U_A,U_B$ are unitaries. Then 
    \begin{align*}
    &(\mathbbm{1}_{n'_A\times m}\otimes \mathbbm{1}_{n'_B\times m})(A_{sa}^{(m)}\otimes B_{tb}^{(m)})\ket{\psi_m}\\
    \approx_\varepsilon&(U_A^*\mathbbm{1}_{n'_A\times n}\otimes U_B^*\mathbbm{1}_{n'_B\times n})(A_{sa}^{(n)}\otimes B_{tb}^{(n)})\ket{\psi_n}.
    \end{align*}
    Take the smallest (or any) $n''\ge\max\{n'_A,n'_B\}$ such that $n''$ is a multiple of $m$. Then
    \begin{align*}
    &(\mathbbm{1}_{n''\times n'_A}U_A^*\mathbbm{1}_{n'_A\times n}\otimes \mathbbm{1}_{n''\times n'_B}U_B^*\mathbbm{1}_{n'_B\times n})(A_{sa}^{(n)}\otimes B_{tb}^{(n)})\ket{\psi_n}\\
    \approx_\varepsilon&(\mathbbm{1}_{n''\times m}\otimes \mathbbm{1}_{n''\times m})(A_{sa}^{(m)}\otimes B_{tb}^{(m)})\ket{\psi_m}\\
    =&(\ket{0}_{A''}\ket{0}_{B''})\otimes (A_{sa}^{(m)}\otimes B_{tb}^{(m)})\ket{\psi_m},
    \end{align*}
    where $\ket{0}_{A''}\in\mathcal{H}_{A''}\cong\mathbb{C}^{n''/m}$, $\ket{0}_{B''}\in\mathcal{H}_{B''}\cong\mathbb{C}^{n''/m}$. It is clear that both $V_A':=\mathbbm{1}_{n''\times n'_A}U_A^*\mathbbm{1}_{n'_A\times n}$ and $V_B':=\mathbbm{1}_{n''\times n'_B}U_B^*\mathbbm{1}_{n'_B\times n}$ are isometries. So $S_n\xhookrightarrow[V_A'\otimes V_B',\ket{0}_{A''}\otimes\ket{0}_{B''}]{\epsilon}S_m$.
\end{proof}

\subsection{Nearly support-preserving strategies}

We introduce the idea of \emph{support-preserving strategies} \cite{David2021}. Given a pure strategy $S=(\ket{\psi},\allowbreak\{A_{sa}\},\{B_{tb}\})$, in the case where $\ket{\psi}$ is not full-rank, the support of $\ket{\psi}$ may or may not be an invariant subspace of the measurement operators. In a \emph{support-preserving} strategy, the measurement operators map the state still inside the support of the state. That is, a quantum strategy $S=(\ket{\psi},\{A_{sa}\},\{B_{tb}\})$ is called \emph{support-preserving} if
\begin{equation*}
        \Supp_A\left((A_{sa} \tensor \mathbbm 1_B)\ket\psi\right) \subseteq \Supp_A(\ket\psi),~
        \Supp_B\left((\mathbbm 1_A\tensor B_{tb}) \ket\psi\right) \subseteq \Supp_B(\ket\psi),
\end{equation*}
holds for all $s,a,t,b$. Alternatively, one also can think of it as the measurement operators being block-diagonal in the Schmidt basis of the state, as what the authors of \cite{paddock2023operatoralgebraic} independently defined therein, which they refer to as ``centrally-supported''. It is given by the following condition:
\begin{equation*}
    [A_{sa},\Pi_A]=[B_{tb},\Pi_B]=0,
\end{equation*}
where $\Pi_A$ and $\Pi_B$ is the projection onto $\Supp_A(\ket\psi)$ and $\Supp_B(\ket\psi)$, respectively. As the latter form is easier to be generalised in the case of robust self-testing, we adopt it to the following definition of \emph{nearly support-preserving} strategies.

\begin{definition}[nearly support-preserving]
    Let $\epsilon\ge0$. A pure strategy $S=(\ket{\psi} \in \mathcal{H}_{A}\otimes \mathcal{H}_B,\{A_{sa}\},\{B_{tb}\})$ is called \emph{$\varepsilon$-support-preserving} if
    \begin{align*}
        \|[\Pi_A,A_{sa}]\|_{\sigma_A}\le\varepsilon,~\|[\Pi_B,B_{tb}]\|_{\sigma_B}\le\varepsilon,
    \end{align*}
    hold for all $s,t,a,b$, where $\Pi_A$ is the projection onto $\Supp_A(\ket\psi)$ (likewise for $\Pi_B$ on Bob), $\sigma_A=\tr_B[\proj{\psi}]$ is the density matrix on Alice (likewise for $\sigma_B$ on Bob), and the state dependent norm is defined as $\|X\|_{\sigma}:=\sqrt{\tr[X^*X\sigma]}$. If further $\varepsilon=0$, $S$ is called \emph{support-preserving} for simplicity.
\end{definition}

Note that $\|[\Pi_A,A_{sa}]\|_{\sigma_A}=\|[\Pi_A,A_{sa}]\otimes \mathbbm{1}_B\ket{\psi}\|=\sqrt{\braket{\psi|(A_{sa}^2-A_{sa}\Pi_AA_{sa})\otimes\mathbbm{1}_B|\psi}}$. This identity is useful in the calculation. Also note that all full-rank strategies are support-preserving by definition.

We will show that  support-preservingness is an invariant property under local dilation. That is, if $S\xhookrightarrow{}\tilde S$, then $S$ is support-preserving if and only if $\tilde S$ is. So this characteristic would not change as we move along `$\xhookrightarrow{}$'. To prove this, the following characterization of near support-preservingness, inspired by \cite[Lemma 4.3]{paddock2023operatoralgebraic}, is needed.

\begin{lemma}
    Let $S=(\ket{\psi},\{A_{sa}\},\{B_{tb}\})$ be a pure strategy. 
    \begin{enumerate}[(a)]
        \item If $S$ is $\varepsilon$-support-preserving, then there exist operators $\hat{A}_{sa}\in\mathcal{H}_B, \hat{B}_{tb}\in\mathcal{H}_A$ such that $A_{sa}\otimes\mathbbm{I}_B\ket{\psi}\approx_\varepsilon\mathbbm{I}_A\otimes\hat{A}_{sa}\ket{\psi}$ and $\mathbbm{I}_A\otimes{B}_{tb}\ket{\psi}\approx_\varepsilon \hat{B}_{tb}\otimes\mathbbm{I}_B\ket{\psi}$ for all $a,s,t,b$.
        \item If there exist operators $\hat{A}_{sa}\in\mathcal{H}_B, \hat{B}_{tb}\in\mathcal{H}_A$ such that $A_{sa}\otimes\mathbbm{I}_B\ket{\psi}\approx_\varepsilon\mathbbm{I}_A\otimes\hat{A}_{sa}\ket{\psi}$ and $\mathbbm{I}_A\otimes{B}_{tb}\ket{\psi}\approx_\varepsilon \hat{B}_{tb}\otimes\mathbbm{I}_B\ket{\psi}$ for all $a,s,t,b$, then $S$ is $2\varepsilon$-support-preserving.
    \end{enumerate}
\label{lem:switchrob}
\end{lemma}
    
\begin{proof}
To prove (a), consider the Schmidt decomposition of the state
    $$\ket{\psi}=\sum_i\lambda_i\ket{e_i}\ket{f_i},\lambda_i>0.
    $$
    Define operators
    \begin{align*}
        &\lambda_{A\to B}:=\sum_{i}\lambda_i\ket{f_i}\!\bra{e_i},\\
        &\lambda^{-1}_{A\to B}:=\sum_{i}\lambda^{-1}_i\ket{f_i}\!\bra{e_i},\\
        &\lambda_{B\to A}:=\sum_{i}\lambda_i\ket{e_i}\!\bra{f_i}=(\lambda_{A\to B})^*,\\
        &\lambda^{-1}_{B\to A}:=\sum_{i}\lambda^{-1}_i\ket{e_i}\!\bra{f_i}=(\lambda^{-1}_{A\to B})^*,
    \end{align*}
    and let $\hat{A}_{sa}:=\lambda_{A\to B}A_{sa}^\intercal\lambda^{-1}_{B\to A}\in\mathcal{H}_B, \hat{B}_{tb}:=\lambda_{B\to A}B_{tb}^\intercal\lambda^{-1}_{A\to B}\in\mathcal{H}_A$, where the transpose are with respect to the bases $\{\ket{e_i}_A\},\{\ket{f_i}_B\}$, respectively. Then
    \begin{align*}
        \mathbbm{I}_A\otimes\hat{A}_{sa}\ket{\psi}=&\mathbbm{1}_A\otimes\lambda_{A\to B}A^\intercal_{sa}\sum_i\ket{e_i}\ket{e_i}\\
        =&\sum_{i,j}\lambda_j\braket{e_j|A_{sa}^\intercal|e_i}\ket{e_i}\ket{f_j}\\
        =&\sum_{i,j}\lambda_j\braket{e_i|A_{sa}|e_j}\ket{e_i}\ket{f_j}\\
        =&\Pi_AA_{sa}\otimes\mathbbm{1}_B\ket{\psi}.
    \end{align*}
    In the last equation we use the identity $\Pi_A=\sum_i\proj{e_i}$.
    So 
    \begin{align*}
        &\|A_{sa}\otimes\mathbbm{I}_B\ket{\psi}-\mathbbm{I}_A\otimes\hat{A}_{sa}\ket{\psi}\|\\
        =&\|A_{sa}\Pi_A\otimes\mathbbm{I}_B\ket{\psi}-\Pi_AA_{sa}\otimes\mathbbm{I}_B\ket{\psi}\|=\|[\Pi_A,A_{sa}]\|_{\sigma_A}.
    \end{align*}
    Then $A_{sa}\otimes\mathbbm{I}_B\ket{\psi}\approx_\varepsilon\mathbbm{I}_A\otimes\hat{A}_{sa}\ket{\psi}$ if $\|[\Pi_A,A_{sa}]\|_{\sigma_A}\le\varepsilon$. The similar argument also works for Bob's operators.

    To prove (b), note that $\|[\Pi_A,A_{sa}]\|_{\sigma_A}=\|\Pi_AA_{sa}\otimes \mathbbm{1}\ket{\psi}-A_{sa}\Pi_A\otimes \mathbbm{1}\ket{\psi}\|$. Then
    \begin{align*}
        \Pi_AA_{sa}\otimes \mathbbm{1}\ket{\psi}&\approx_\varepsilon\Pi_A\otimes \hat{A}_{sa}\ket{\psi}\\
        &=\mathbbm{1}\otimes \hat{A}_{sa}\ket{\psi}\\
        &\approx_\varepsilon A_{sa}\otimes \mathbbm{1}\ket{\psi}=A_{sa}\Pi_A\otimes \mathbbm{1}\ket{\psi}.
    \end{align*}
    So $\Pi_AA_{sa}\otimes \mathbbm{1}\ket{\psi}\approx_{2\varepsilon}A_{sa}\Pi_A\otimes \mathbbm{1}\ket{\psi}$. The similar argument also works for Bob's operators.
\end{proof}

The invariance of support-preservingness under local dilation can be stated as follows:

\begin{proposition}
\label{prop:invariantSupp}
Let $S$ and $\tilde S$ be two pure strategies.
\begin{enumerate}[(a)]
    \item If $S\xhookrightarrow{} \tilde{S}$, then ${S}$ is $\varepsilon$-support-preserving if and only if $\tilde S$ is $\varepsilon$-support-preserving.
    \item If $S\xhookrightarrow{\varepsilon'} \tilde{S}$, then $\tilde{S}$ being $\varepsilon$-support-preserving implies that $S$ is $(4\varepsilon'+2\varepsilon)$-support-preserving, and ${S}$ being $\varepsilon$-support-preserving implies that $\tilde S$ is $(4\varepsilon'+2\varepsilon)$-support-preserving.
\end{enumerate}
\end{proposition}

\begin{proof}
Let $V_A\otimes V_B$ be the local isometry and $\ket{\aux}$ be the auxiliary state in the exact/near local-dilation.

To prove (a), note that $V_A\Pi_A V_A^*=\tilde{\Pi}_A\otimes\Pi_{A'}$, where $\tilde{\Pi}_A$ and $\Pi_{A'}$ are projections onto $\Supp_A{\ket{\tilde\psi}}$ and $\Supp_A{\ket{\aux}}$, respectively. Then
\begin{align*}
    \|[\Pi_A,A_{sa}]\|_{\sigma_A}^2=&\braket{\psi|(A_{sa}^2-A_{sa}\Pi_AA_{sa})\otimes\mathbbm{1}_B|\psi}\\
    =&\braket{\psi|A_{sa}V_A^*V_A(A_{sa}-\Pi_AV_A^*V_AA_{sa})\otimes V_B^*V_B|\psi}\\
    =&\braket{\tilde\psi,\aux|[(\tilde{A}_{sa}\otimes\mathbbm{1}_{A'})(\tilde{A}_{sa}\otimes\mathbbm{1}_{A'}-V_A\Pi_AV_A^*(\tilde{A}_{sa}\otimes\mathbbm{1}_{A'}))]\otimes \mathbbm{1}_{\tilde B,\hat{B}}|\tilde \psi,\aux}\\
    =&\braket{\tilde\psi,\aux|(\tilde{A}^2_{sa}-\tilde{A}_{sa}\tilde{\Pi}_A\tilde{A}_{sa})\otimes {\Pi}_{A'}\otimes\mathbbm{1}_{\tilde B,\hat{B}}|\tilde \psi,\aux}\\
    =&\braket{\tilde\psi|(\tilde{A}^2_{sa}-\tilde{A}_{sa}\tilde{\Pi}_A\tilde{A}_{sa})\otimes\mathbbm{1}_{\tilde B}|\tilde \psi}=\|[\tilde{\Pi}_A,\tilde{A}_{sa}]\|_{\sigma_{\tilde{A}}}^2.
\end{align*}

So $\|[\Pi_A,A_{sa}]\|_{\sigma_A}\le\varepsilon$ if and only if $\|[\tilde{\Pi}_A,\tilde{A}_{sa}]\|_{\sigma_{\tilde{A}}}\le\varepsilon$. The similar argument also works for Bob's operators.

In (b), we first prove the first implication.

Since $\tilde{S}$ is $\varepsilon$-support-preserving, by Lemma \ref{lem:switchrob} there exist $\hat{\tilde{A}}_{sa}$ such that $\tilde{A}_{sa}\otimes \mathbbm{1}\ket{\tilde{\psi}}\approx_{\varepsilon}\mathbbm{1}\otimes \hat{\tilde{A}}_{sa}\ket{\tilde{\psi}}$. From the near local dilation, we have that 
\begin{align}
    (V_AV_A^*\otimes V_BV_B^*)(\tilde{A}_{sa}\otimes\mathbbm{1}_{\tilde B}\ket{\tilde{\psi}}\otimes\ket{\aux})\approx_\varepsilon(V_A\otimes V_B)(A_{sa}\otimes\mathbbm{1}_B)\ket{\psi}.
    \label{eq:vv*}
\end{align}

Consider the operator $\hat{A}_{sa}:=V_B^*(\hat{\tilde{A}}_{sa}\otimes \mathbbm{1}_{\hat{B}})V_B$, then
\begin{align}
    V_A\otimes V_B(\mathbbm{1}_A\otimes \hat{A}_{sa}\ket{\psi})&=V_A\otimes V_B(\mathbbm{1}_A\otimes V_B^*(\hat{\tilde{A}}_{sa}\otimes \mathbbm{1}_{\hat{B}})V_B\ket{\psi})\nonumber\\
    &=(V_AV_A^*\otimes V_BV_B^*(\hat{\tilde{A}}_{sa}\otimes \mathbbm{1}_{\hat{B}}))(V_A\otimes V_B)\ket{\psi}\nonumber\\
    &\approx_{\varepsilon'}(V_AV_A^*\otimes V_BV_B^*\hat{\tilde{A}}_{sa})(\ket{\tilde{\psi}}\otimes\ket{\aux})\nonumber\\
    &\approx_{\varepsilon}(V_AV_A^*\otimes V_BV_B^*)((\tilde{A}_{sa}\otimes\mathbbm{1}_{\tilde B})\ket{\tilde{\psi}}\otimes\ket{\aux})\nonumber\\
    &\approx_{\varepsilon'} V_A\otimes V_B({A}_{sa}\otimes \mathbbm{1}_B)\ket{\psi}.
\end{align}
So $V_A\otimes V_B(\mathbbm{1}_A\otimes \hat{A}_{sa}\ket{\psi})\approx_{2\varepsilon'+\varepsilon}V_A\otimes V_B({A}_{sa}\otimes \mathbbm{1}_B\ket{\psi})$, which implies $(\mathbbm{1}_A\otimes \hat{A}_{sa}\ket{\psi})\approx_{2\varepsilon'+\varepsilon}({A}_{sa}\otimes \mathbbm{1}_B\ket{\psi})$. By Lemma \ref{lem:switchrob}, $\|[\Pi_A,A_{sa}]\|_{\sigma_A}\le{4\varepsilon'+2\varepsilon}$ for all $a,s$. The similar argument holds for Bob's operators. So we conclude that $S$ is $(4\varepsilon'+2\varepsilon)$-support-preserving. 

For the second implication of (b), given the existence of $\hat{A}_{sa}$, consider $\hat{\tilde A}_{sa}:=V_B\hat{A}_{sa}V_B^*$, then
\begin{align*}
    \mathbbm{1}_{\tilde A}\otimes\hat{\tilde A}_{sa}\ket{\tilde\psi}\ket{\aux}&=\mathbbm{1}_{{\tilde A},A'}\otimes V_B\hat{ A}_{sa}V_B^*(\ket{\tilde\psi}\ket{\aux})\\
    &\approx_{\varepsilon'}V_A\otimes V_B\hat{A}_{sa}\ket{\psi}\\
    &\approx_{\varepsilon}V_AA_{sa}\otimes V_B\ket{\psi}\\
    &\approx_{\varepsilon'}\tilde{A}_{sa}\otimes\mathbbm{1}_{\tilde B}\ket{\tilde\psi}\ket{\aux}.
\end{align*}
So by Lemma \ref{lem:switchrob}, $\|[\tilde{\Pi}_A,\tilde{A}_{sa}]\|_{\sigma_{\tilde{A}}}=\|[\tilde{\Pi}_A\otimes\Pi_{A'},\tilde{A}_{sa}\otimes\mathbbm{1}_{A'}]\|_{\sigma_{\tilde{A},A'}}\le{4\varepsilon'+2\varepsilon}$ for all $a,s$. The similar argument holds for Bob's operators. So we conclude that $\tilde S$ is $(4\varepsilon'+2\varepsilon)$-support-preserving. 
\end{proof}

(It has come to our attention that the exact ($\varepsilon=0$) case of the "only if" direction of part (a) of Proposition \ref{prop:invariantSupp} was independently developed by \cite[Proposition 4.6]{paddock2023operatoralgebraic}.)

\subsection{Nearly projective strategies}

We introduce the definition of \emph{nearly projective} strategies. This notion quantifies `how projective a strategy is on its state'.

\begin{definition}[nearly projective]
    Let $\varepsilon\ge0$. A strategy $S=(\ket{\psi} \in \mathcal{H}_{A}\otimes \mathcal{H}_B,\allowbreak\{A_{sa}\},\{B_{tb}\})$ is called \emph{$\varepsilon$-projective} if
    \begin{align*}
        &\braket{\mathbbm{1}_A-A_{sa},A_{sa}}_{\sigma_A}\le\varepsilon^2,\\
        &\braket{\mathbbm{1}_B-B_{tb},A_{tb}}_{\sigma_B}\le\varepsilon^2
    \end{align*}
    hold for all $s,t,a,b$. Here $\braket{X,Y}_{\sigma}:=\tr[X^*Y\sigma]$.
\end{definition}

Note that $\braket{\mathbbm{1}_A-A_{sa},A_{sa}}_{\sigma_A}=\braket{\psi|(\mathbbm{1}_A-A_{sa})A_{sa}\otimes \mathbbm{1}_B|\psi}$, and this identity is useful in some calculations. Also note that being $0$-projective does not necessarily imply being projective: a non-projective strategy might be only non-projective outside of the support of the state, so it could be $0$-projective. But for full-rank strategies, being projective and $0$-projective are equivalent.

The projectiveness of strategies is another invariant property under local dilation. Namely,

\begin{proposition}
Let $S$ and $\tilde S$ be two pure strategies.
\begin{enumerate}[(a)]
    \item If $S\xhookrightarrow{} \tilde{S}$, then ${S}$ is $\varepsilon$-projective if and only if $\tilde S$ is $\varepsilon$-projective.
    \item If $S\xhookrightarrow{\varepsilon'} \tilde{S}$, then $\tilde{S}$ being $\varepsilon$-projective implies that $S$ is $(\sqrt{3\varepsilon'}+\varepsilon)$-projective, and ${S}$ being $\varepsilon$-projective implies that $\tilde S$ is $(\sqrt{3\varepsilon'}+\varepsilon)$-projective.
\end{enumerate}
    \label{prop:invariantProj}
\end{proposition}

\begin{proof}
Since (a) is an implication of (b) (by taking $\varepsilon'=0$), we only need to prove (b).

Given that $S\xhookrightarrow{\varepsilon'}\tilde S$, there exists a local isometry and auxiliary state such that
    \begin{align}
    &V[A_{sa}\otimes \mathbbm{1}\ket{\psi}]\approx_{\varepsilon'}(\tilde{A}_{sa}\otimes \mathbbm{1}\ket{\tilde{\psi}})\otimes\ket{\aux}\label{eq:1},\forall \ a,s\\
    &V[\ket{\psi}]\approx_{\varepsilon'}\ket{\tilde{\psi}}\otimes\ket{\aux}\label{eq:2}
\end{align}

$\eqref{eq:2}-\eqref{eq:1}$:
\begin{align}
    V[(\mathbbm{1}-A_{sa})\otimes \mathbbm{1}\ket{\psi}]\approx_{2\varepsilon'}((\mathbbm{1}-\tilde{A}_{sa})\otimes \mathbbm{1}\ket{\tilde{\psi}})\otimes\ket{\aux}\label{eq:3}
\end{align}

Then the inner product of \eqref{eq:1} and \eqref{eq:3}:
\begin{align}
    \braket{\psi|(A_{sa}-A_{sa}^2)\otimes \mathbbm{1}|\psi}\approx_{3\varepsilon'}\braket{\tilde{\psi}|(\tilde{A}_{sa}-\tilde{A}_{sa}^2)\otimes \mathbbm{1}|\tilde{\psi}}\nonumber.
\end{align}

Note that both sides are real positive numbers, then 
\begin{align*}
    &|\sqrt{\braket{\psi|(A_{sa}-A_{sa}^2)\otimes \mathbbm{1}|\psi}}-\sqrt{\braket{\tilde{\psi}|(\tilde{A}_{sa}-\tilde{A}_{sa}^2)\otimes \mathbbm{1}|\tilde{\psi}}}|\\
    \le&|\sqrt{\braket{\psi|(A_{sa}-A_{sa}^2)\otimes \mathbbm{1}|\psi}}+\sqrt{\braket{\tilde{\psi}|(\tilde{A}_{sa}-\tilde{A}_{sa}^2)\otimes \mathbbm{1}|\tilde{\psi}}}|\\
    =&\frac{|\braket{\psi|(A_{sa}-A_{sa}^2)\otimes \mathbbm{1}|\psi}-\braket{\tilde{\psi}|(\tilde{A}_{sa}-\tilde{A}_{sa}^2)\otimes \mathbbm{1}|\tilde{\psi}}|}{|\sqrt{\braket{\psi|(A_{sa}-A_{sa}^2)\otimes \mathbbm{1}|\psi}}-\sqrt{\braket{\tilde{\psi}|(\tilde{A}_{sa}-\tilde{A}_{sa}^2)\otimes \mathbbm{1}|\tilde{\psi}}}|}\\
    \implies&|\sqrt{\braket{\psi|(A_{sa}-A_{sa}^2)\otimes \mathbbm{1}|\psi}}-\sqrt{\braket{\tilde{\psi}|(\tilde{A}_{sa}-\tilde{A}_{sa}^2)\otimes \mathbbm{1}|\tilde{\psi}}}|\le\sqrt{3\varepsilon'}.
\end{align*}

Then the two implications in (b) follows immediately.
\end{proof}

\subsection{Restrictions of nonlocal strategies}

When a pure strategy $S$ is not full-rank, that is, the Schmidt rank of the state is strictly smaller than the local dimension, a projective/non-projective strategy might be non-projective/projective on the support. This naturally leads us to study the behaviour of the measurements on the support of the state. To this end, we define the \emph{restriction} of a strategy as follows.

\begin{definition}[Restriction of a strategy]
    Let $S=(\ket{\psi},\{A_{sa}\},\{B_{tb}\})$ be a strategy. Consider the Schmidt decomposition of $\ket{\psi}$,
	\begin{equation*}
		\ket{\psi} = \sum_{i = 0}^{d-1} \lambda_i \ket{e_i}\ket{f_i}.
	\end{equation*}
	We define the isometries
	\begin{equation*}
	    U_A = \sum_{i = 0}^{d-1} \ket{e_i}_A\!\bra{i} \qquad \text{and}\qquad U_B =  \sum_{i = 0}^{d-1} \ket{f_i}_B\!\bra{i}
	\end{equation*}
    The \emph{restriction} of $S$ is the strategy $S_{\text{res}}=(\ket{\psi'},\{A_{sa}'\},\{B_{tb}'\})$, where
	\begin{align*}
	    A_{sa}' &= U_A^*  A_{sa} U_A\\  
	    B_{tb}' &= U_B^*  B_{tb} U_B\\
	    \ket{\psi'} &=\sum_{i = 0}^{d-1} \lambda_i \ket{i}\ket{i}= U_A^*\otimes U_B^* \ket{\psi}.
	\end{align*}
	It is indeed a well-defined POVM full-rank strategy. 
 \label{def:res}
\end{definition}

Note that the projections on the supports of $S$ can be written as $\Pi_A=U_AU_A^*,\Pi_B=U_BU_B^*$.

We will now see that if a non-full-rank strategy $S$ is exactly/nearly support-preserving, then $S$ and its restriction $S_{\operatorname{res}}$ (defined as in Definition \ref{def:res}) can be mutually exactly/nearly local-dilated. 

\begin{proposition}
    If a pure strategy $S$ is $\varepsilon$-support-preserving, then $S_{\operatorname{res}}\xhookrightarrow{\varepsilon}S$ and $S\xhookrightarrow{\varepsilon}S_{\operatorname{res}}$, where $S_{\operatorname{res}}$ is the restriction of $S$.
    \label{prop:restrictrob}
\end{proposition}

\begin{proof}
    We show that $S_{\text{res}}\xhookrightarrow{\varepsilon}S$ with a separable auxiliary state, then $S\xhookrightarrow{\varepsilon}S_{\operatorname{res}}$ follows from Proposition \ref{prop:otherdirection}.
    
    Consider isometries $U_A,U_B$ in Definition \ref{def:res}, and recall that $U_AU_A^*=\Pi_A,U_BU_B^*=\Pi_B$. Then 
    \begin{align*}
	U_A\otimes U_B ( A'_{st} \otimes \mathbbm{1}_B) \ket{ \psi'} &= U_AU_A^* A_{sa} U_AU_A^* \otimes U_BU_B^* \ket{ \psi}\\
        &=\Pi_A A_{sa} \Pi_A \otimes \Pi_B \ket{ \psi}\\
        &\approx_\varepsilon A_{sa} \Pi_A \otimes \Pi_B \ket{ \psi}\\
	    &=A_{sa} \otimes \mathbbm{1}_B \ket{ \psi}.
    \end{align*}
A similar argument holds for Bob's operators. So $ S_{\operatorname{res}}\xhookrightarrow{\varepsilon} S$.
\end{proof}

As a consequence of Proposition \ref{prop:restrictrob}, if a game $G$ pure self-tests a support-preserving canonical strategy $\tilde S$, $G$ also pure self-tests its restriction $\tilde S_{\text{res}}$. In other words, in this case, we can always take a full-rank canonical strategy.

In general, a projective strategy might become non-projective under restriction. Here, we show that the other way around can never happen: whenever a restriction is projective, the original strategy must be both projective and support-preserving.

\begin{theorem}
    Let $S=(\ket{\psi},\{A_{sa}\},\{B_{tb}\})$ be a pure strategy and $S_{\operatorname{res}}$ be its restriction. 
    \begin{enumerate}[(a)]
        \item If $S$ is $\varepsilon_1$-support-preserving and $\varepsilon_2$-projective, then $S_{\operatorname{res}}$ is $(\varepsilon_1+\varepsilon_2)$-projective.
        \item If $S_{\operatorname{res}}$ is $\varepsilon_3$-projective, then $S$ is $\varepsilon_3$-support-preserving and $\varepsilon_3$-projective.
    \end{enumerate}
    \label{thm:resProj}
\end{theorem}

\begin{proof}
    We prove for Alice's side, and the same argument works also for Bob. By definition,
    \begin{align*}
        \|[\Pi_A,A_{sa}]\|_{\sigma_A}^2=&\braket{\psi|(A_{sa}\Pi_A-\Pi_A A_{sa})(\Pi_A A_{sa}-A_{sa}\Pi_A)\otimes\mathbbm{1}|\psi}\\
        =&\braket{\psi|(A_{sa}\Pi_A A_{sa}-A_{sa}\Pi_A A_{sa}\Pi_A-\Pi_A A_{sa}\Pi_A A_{sa}+\Pi_AA^2_{sa}\Pi_A)\otimes\mathbbm{1}|\psi}\\
        =&\braket{\psi|(A_{sa}^2-A_{sa}\Pi_A A_{sa})\otimes\mathbbm{1}|\psi}.
    \end{align*}

    On the other hand, recall that $S_{\text{res}}=(\ket{\psi'},\{A_{sa}'\},\{B_{tb}'\})$ where $\ket{\psi'}=U_A^*\otimes U_B^*\ket{\psi}$, $U'_A=U_A^*A_{sa}U_A$, and $U_A$ satisfies $U_AU_A^*=\Pi_A$. So 
    \begin{align*}
        \braket{\mathbbm{1}-A'_{sa},A'_{sa}}_{\sigma'_A}=&\braket{\psi|(U_A\otimes U_B)((\mathbbm{1}-U_A^*A_{sa}U_A)U_A^*A_{sa}U_A\otimes\mathbbm{1})(U^*_A\otimes U^*_B)|\psi}\\
        =&\braket{\psi|(\Pi_AA_{sa}\Pi_A-\Pi_AA_{sa}\Pi_AA_{sa}\Pi_A)\otimes\mathbbm{1}|\psi}\\
        =&\braket{\psi|(A_{sa}-A_{sa}\Pi_AA_{sa})\otimes\mathbbm{1}|\psi}.
    \end{align*}
    Therefore
    \begin{align*}
        \braket{\mathbbm{1}-A'_{sa},A'_{sa}}_{\sigma'_A}-\|[\Pi_A,A_{sa}]\|_{\sigma_A}^2=&\braket{\psi|(A_{sa}-A_{sa}^2)\otimes\mathbbm{1}|\psi}\\
        =&\braket{\psi|((\mathbbm{1}-A_{sa})A_{sa})\otimes\mathbbm{1}|\psi}\\
        =&\braket{\mathbbm{1}-A_{sa},A_{sa}}_{\sigma_A}.
    \end{align*}

    Then (a) is clear. For (b), note that both $\braket{\mathbbm{1}-A_{sa},A_{sa}}_{\sigma_A}$ and $\|[\Pi_A,A_{sa}]\|_{\sigma_A}^2$ are positive. So if $\braket{\mathbbm{1}-A'_{sa},A'_{sa}}_{\sigma'_A}\le\varepsilon_3$ then $\braket{\mathbbm{1}-A_{sa},A_{sa}}_{\sigma_A}\le\varepsilon_3$ and $\|[\Pi_A,A_{sa}]\|_{\sigma_A}^2\le\varepsilon_3$.
\end{proof}

\begin{corollary}\label{cor:restrictionprojective}
    The restriction $S_{\operatorname{res}}$ is projective if and only $S$ is support-preserving and $0$-projective (i.e. projective on the support of the state). 
\end{corollary}

\subsection{Naimark dilation of nonlocal strategies}

The Naimark dilation theorem provides an essential framework for characterizing POVMs, having significant influence not only in this study but also in the broader domains of operator theory and quantum information theory. For a given (finite) set of POVMs, it Naimark dilation can be defined as the following:

\begin{definition}
    Let $\{R_{ij}\}_{j=1}^{m_i}$, $1\leq i \leq n$, be a family of POVMs on $\mathcal{H}$. $(\{P_{ij}\}_{j=1}^{m_i},V)$ is called a \emph{Naimark dilation} of $\{R_{ij}\}_{j=1}^{m_i}$, if $\{P_{ij}\}_{j=1}^{m_i}$ is a family of PVMs on $\mathcal{H}'$, $V:\mathcal{H}\to\mathcal{H}'$, and $R_{ij}=V^*P_{ij}V$ for all $i,j$.
    \label{def:Naimark}
\end{definition}

Within the range of its diverse forms, we here introduce a specific variant of the Naimark dilation theorem, particularly suited clearly to handle finite-dimensional POVMs. This construction can also be found in \cite[Proposition 9.6 and Theorem 9.8]{PaulsenNaimark}. It's important to note that while this construction serves to illuminate the intuition behind Naimark dilations, the results we present later in the paper are not tied to this specific example. Our general theorem applies to any Naimark dilation, regardless of its particular structure.

\begin{construction}
    Let $\{R_{ij}\}_{j=1}^{m_i}$, $1\leq i \leq n$, be a family of POVM's on $\mathcal{H}$. Construct projective measurements $\{P_{ij}\}_{j=1}^{m_i}$ and an isometry $V$ such that $R_{ij}=V^*P_{ij}V$ as follows:

  For a single POVM $\{R_j\}_{j=1}^{m}$, define $P_j=\mathbbm{1}_{B(\mathcal{H})}\otimes e_{j}e_j^*$ in $B(\mathcal{H}\otimes \mathbb{C}^m)$, $1 \leq j \leq m$ and $V:\mathcal{H}\to \mathcal{H} \otimes \mathbb{C}^m$, $\varphi \mapsto \sum_{j=1}^m \sqrt{R_j}\varphi \otimes e_j$. 
  
    Now assume we did the construction for $n$ POVM's. Take a family $\{R_{ij}\}_{j=1}^{m_i}$, $1\leq i \leq n+1$. Then we have PVM's $\{Q_{ij}\}_{j=1}^{m_i}$, $1\leq i \leq n$, on $\mathcal{H}\otimes \mathbb{C}^{m_1}\otimes\cdots\otimes\mathbb{C}^{m_n}$ and an isometry $V_1:\mathcal{H}\to \mathcal{H}\otimes \mathbb{C}^{m_1}\otimes\cdots\otimes\mathbb{C}^{m_n}$ such that $R_{ij}=V_1^*Q_{ij}V_1$. From $\{R_{n+1j}\}_{j=1}^{m_{n+1}}$, we get a POVM on $\mathcal{H}\otimes \mathbb{C}^{m_1}\otimes\cdots\otimes\mathbb{C}^{m_n}$ by taking 
\begin{align*}
    \tilde{R}_{n+1j}&=V_1R_{n+1j}V_1^*\quad\intertext{for $j\neq 1$ and }\,\tilde{R}_{(n+1)1}&=V_1R_{(n+1)1}V_1^* + (\mathbbm{1}_{B(\mathcal{H}\otimes \mathbb{C}^{m_1}\otimes\cdots\otimes\mathbb{C}^{m_n})}-V_1V_1^*).
\end{align*} 
Similar to the single POVM case, we define 
\begin{align*}
    P_{n+1j}=\mathbbm{1}_{B(\mathcal{H}\otimes \mathbb{C}^{m_1}\otimes\cdots\otimes\mathbb{C}^{m_n})}\otimes e_{j}e_j^*
\end{align*}
in $B(\mathcal{H}\otimes \mathbb{C}^{m_1}\otimes\cdots\otimes\mathbb{C}^{m_{n+1}})$, $1 \leq j \leq m$ and 
\begin{align*}
    V_2:\mathcal{H}\otimes \mathbb{C}^{m_1}\otimes\cdots\otimes\mathbb{C}^{m_n} \to \mathcal{H}\otimes \mathbb{C}^{m_1}\otimes\cdots\otimes\mathbb{C}^{m_{n+1}}, \varphi \mapsto \sum_{j=1}^{m_{n+1}} \sqrt{\tilde{R}_{n+1j}}\varphi \otimes e_j. 
\end{align*}
For $1\leq i \leq n$, we let
 \begin{align*}
    P_{ij}=V_2Q_{ij}V_2^*\quad\text{for $i\neq 1$ and }\,P_{i1}=V_2Q_{i1}V_2^* + (\mathbbm{1}_{B(\mathcal{H}\otimes \mathbb{C}^{m_1}\otimes\cdots\otimes\mathbb{C}^{m_{n+1}})}-V_2V_2^*).
 \end{align*}
Finally, we define $V:=V_2\circ V_1$. 
    \label{con:Naimark}
\end{construction}

The following theorem shows that this construction  gives us a valid projective measurement. 

\begin{theorem}\cite[Theorem 9.8]{PaulsenNaimark}
 Let $\{R_{ij}\}_{j=1}^{m_i}$, $1\leq i \leq n$, be a family of POVM's on $\mathcal{H}$. Then the $\{P_{ij}\}_{j=1}^{m_i}$ given in Construction \ref{con:Naimark} are projective measurements, and we have $R_{ij}=V^*P_{ij}V$ for the isometry $V$ given therein.
\end{theorem}

In our subsequent analysis, we show that our results hold for all Naimark dilations, thus are not limited by the specific details of this construction.

Given the Naimark dilation of multiple POVMs, one can talk about the Naimark dilation of a strategy:

\begin{definition}[Naimark dilation of quantum strategies]
    Given a pure strategy $S=(\ket{\psi},\{A_{sa}\},\{B_{tb}\})$, a PVM strategy $S_{\operatorname{Naimark}}=(V_A\otimes V_B\ket{\psi},\{P_{sa}\},\{Q_{tb}\})$ is called a Naimark dilation of $S$, if $(\{P_{sa}\},V_A)$ is a Naimark dilation of $\{A_{sa}\}$, and $(\{Q_{tb}\},V_B)$ is a Naimark dilation of $\{B_{tb}\}$. 
\end{definition}

And not surprisingly, they generate the same statistics:
\begin{lemma}
    Any pure strategy gives the same correlation as its Naimark dilations.\label{lem:Naimarksamecorrelation}
\end{lemma}

\begin{proof}
    Let $S=(\ket{\psi},\{A_{sa}\},\{B_{tb}\})$ and $S_{\operatorname{Naimark}}=(V_A\otimes V_B\ket{\psi},\{P_{sa}\},\{Q_{tb}\})$. Using $A_{sa}=V_A^*P_{sa}V_A,B_{tb}=V_B^*Q_{tb}V_B$,we get
    $$
    \braket{\psi|A_{sa}\otimes B_{tb}|\psi}=\braket{\psi|V^*VA_{sa}\otimes B_{tb}V^*V|\psi}=\braket{V\psi|P_{sa}\otimes Q_{tb}|V\psi},
    $$
    where $V=V_A\otimes V_B$.
\end{proof}

As an analog of Proposition \ref{prop:restrictrob}, we will show that $S$ and $S_{\operatorname{Naimark}}$ are mutually locally dilated if $S$ is projective. To prove this, we need the following lemma:

\begin{lemma}
\label{lem:naimarklocaldilationrob}
Let $\{R_{ij}\}_{j=1}^m$, $1 \leq i \leq n$, be a collection of POVM's on $\mathcal{H}$, $\sigma$ be a density matrix on $\mathcal{H}$, and $\ket{\psi}$ be a purification of $\sigma$. Then any Naimark dilation $(\{P_{ij}\},V)$ of $\{R_{ij}\}$ satisfies 
\begin{align*}
    \|VR_{ij}\otimes \mathbbm{1}\ket{\psi}- P_{ij}V\otimes \mathbbm{1}\ket{\psi}\|^2=\braket{\psi|(\mathbbm{1}-R_{ij})R_{ij}\otimes \mathbbm{1}_B|\psi}.
\end{align*}
\end{lemma}

\begin{proof}
Using $V^*P_{ij}V=R_{ij}$, we get
\begin{align*}
    &\|VR_{ij}\otimes \mathbbm{1}\ket{\psi}- P_{ij}V\otimes \mathbbm{1}\ket{\psi}\|^2\\
    =&\braket{\psi|(R_{ij}V^*VR_{ij}+V^*P_{ij}P_{ij}V-V^*P_{ij}VR_{ij}-R_{ij}V^*P_{ij}V)\otimes\mathbbm{1}|\psi}\\
    =&\braket{\psi|(R_{ij}^2+R_{ij}-R^2_{ij}-R_{ij}^2)\otimes\mathbbm{1}|\psi}\\
    =&\braket{\psi|(\mathbbm{1}-R_{ij})R_{ij}\otimes \mathbbm{1}_B|\psi}.
\end{align*}
\end{proof}

Applying Lemma \ref{lem:naimarklocaldilationrob} in the context of nonlocal strategies, we have the following proposition:

\begin{proposition}
    If a pure strategy $S$ is $\varepsilon$-projective, then $S\xhookrightarrow{\varepsilon}S_{\operatorname{Naimark}}$ and $S_{\operatorname{Naimark}}\xhookrightarrow{\varepsilon}S$, where $S_{\operatorname{Naimark}}$ is any Naimark dilation of $S$.
\label{prop:Naimarkrob}
\end{proposition}

\begin{proof}
    It is clear from Lemma \ref{lem:naimarklocaldilationrob} that $S\xhookrightarrow[V_A\otimes V_B]{\varepsilon}S_{\operatorname{Naimark}}$, where $V_A,V_B$ are isometries given in Definition \ref{def:Naimark}. Then $S_{\operatorname{Naimark}}\xhookrightarrow{\varepsilon}S$ follows from Proposition \ref{prop:otherdirection}.
\end{proof}

We now show that if a Naimark dilation of a strategy is support-preserving, then the original one must be both projective and support-preserving (an analog of Theorem \ref{thm:resProj}).

\begin{theorem}
    Let $S=(\ket{\psi},\{A_{sa}\},\{B_{tb}\})$ be a pure strategy and $S_{\operatorname{Naimark}}$ be any Naimark dilation of $S$.
    \begin{enumerate}[(a)]
        \item If $S$ is $\varepsilon_1$-support-preserving and $\varepsilon_2$-projective, then $S_{\operatorname{Naimark}}$ is $(\varepsilon_1+\varepsilon_2)$-support-preserving.
        \item If $S_{\operatorname{Naimark}}$ is $\varepsilon_3$-support-preserving, then $S$ is $\varepsilon_3$-support-preserving and $\varepsilon_3$-projective.
    \end{enumerate}
    \label{thm:NaiSupp}
\end{theorem}

\begin{proof}
    We prove for Alice’s side, and the same argument works also for Bob. Let $\Pi$ be the projection on the support of $\ket{\psi}$ on Alice's side. Let $(\{P_{sa}\},V)$ be the Naimark dilation of $\{A_{sa}\}$. Note that 
\begin{align*}
    \|[V\Pi V^*,P_{sa}]\|^2_{V\sigma V^*}=&\|P_{sa}V\Pi V^*V\otimes\mathbbm{1}\ket{\psi}-V\Pi V^*P_{sa}V\otimes\mathbbm{1}\ket{\psi}\|^2\\
    =&\braket{\psi|A_{sa}\otimes\mathbbm{1}|\psi}-\braket{\psi|A_{sa}\Pi A_{sa}\otimes\mathbbm{1}|\psi}.
\end{align*}
And
\begin{align*}
    \|[\Pi,A_{sa}]\|^2_{\sigma}=&\|\Pi A_{sa}\otimes\mathbbm{1}\ket{\psi}-A_{sa}\Pi\otimes\mathbbm{1}\ket{\psi}\|^2\\
    =&\braket{\psi|A^2_{sa}\otimes\mathbbm{1}|\psi}-\braket{\psi|A_{sa}\Pi A_{sa}\otimes\mathbbm{1}|\psi}.
\end{align*}
So
\begin{align*}
    &\|[V\Pi V^*,P_{sa}]\|^2_{V\sigma V^*}-\braket{\mathbbm{1}- A_{sa},A_{sa}}_{\sigma}\\
    =&\braket{\psi|A^2_{sa}\otimes\mathbbm{1}|\psi}-\braket{\psi|A_{sa}\Pi A_{sa}\otimes\mathbbm{1}|\psi}\\
    =&\|[\Pi,A_{sa}]\|^2_{\sigma}.
\end{align*}
    
    Then (a) is clear. For (b), note that both $\braket{\mathbbm{1}- A_{sa},A_{sa}}_{\sigma}$ and $\|[\Pi,A_{sa}]\|^2_{\sigma}$ are positive. So $S_{\operatorname{Naimark}}$ being $\varepsilon_3$-support-preserving implies that $S$ is $\varepsilon_3$-projective and $\varepsilon_3$-support-preserving.
\end{proof}

\begin{corollary}\label{cor:Naimarksupportpreserving}
    The Naimark dilation $S_{\operatorname{Naimark}}$ is support-preserving if and only if $S$ is support-preserving and $0$-projective (i.e. projective on the support of the state). 
\end{corollary}
	\section{Lifting Assumptions}\label{sec:lifting}

In this section, we aim to lift the assumptions that are commonly made in the literature. Specifically, we will establish the following theorem, which is our main result.

\begin{theorem}[Main theorem]
\label{thm:bigtheorem}
    Let $\tilde S$ be a pure strategy that is optimal for a nonlocal game $G$. Then the following two implications hold:
    \begin{enumerate}[(a)]
        \item If $\tilde S$ is full-rank and $G$ is a robust pure PVM self-test for $\tilde S$, then $G$ is a robust assumption-free self-test for $\tilde S$.
        \item If $\tilde S$ is projective and $G$ is a robust pure full-rank self-test for $\tilde S$, then $G$ is a robust assumption-free self-test for $\tilde S$.
    \end{enumerate}
\end{theorem}

We will prove Theorem \ref{thm:bigtheorem} $(a)$ and $(b)$ both in two steps. For part $(a)$, we first show how to get rid of the PVM assumption in the following subsections. Similarly, we show that we can lift the full-rank assumption in part $(b)$ in Subsection \ref{sec:rank}. Then, for both Theorem \ref{thm:bigtheorem} $(a)$ and $(b)$, we get from pure to mixed states in Subsection \ref{sec:mainres}. 

Theorem \ref{thm:bigtheorem} tells us that, if a strategy $\tilde S$ is proved self-tested with the assumption of either full-rank states or projective measurements (for the arbitrary strategy), and $\tilde S$ itself is full-rank and projective, then those assumptions can be lifted for free. Moreover, we show that for $\tilde S$ to be assumption-free self-tested, it has to be both support-preserving and $0$-projective:
\begin{theorem}
    \label{thm:originalpartC}
    If $G$ is an assumption-free self-test for $\tilde S$, then $\tilde S$ must be $0$-projective and support-preserving. Moreover, $G$ is an assumption-free self-test for the restriction of $\tilde S$.
\end{theorem}
\begin{proof}
    Take any pure strategy $S$ that wins $G$ optimally, then so does its Naimark dilation $S_{\operatorname{Naimark}}$ and its restriction $S_{\operatorname{res}}$. Since $G$ is an assumption-free self-test, it is also a pure self-test. So $S_{\operatorname{res}}\xhookrightarrow{}\tilde S$, which implies that $\tilde S$ is support-preserving using Proposition \ref{prop:invariantSupp} and the fact that $S_{\operatorname{res}}$ is support-preserving. Similarly, $S_{\operatorname{Naimark}}\xhookrightarrow{}\tilde S$ implies $\tilde S$ to be $0$-projective by Proposition \ref{prop:invariantProj} and the fact that $S_{\operatorname{Naimark}}$ is projective.

    Since $\tilde S$ is $0$-projective and support-preserving, $\tilde S\xhookrightarrow{}\tilde S_{\text{res}}$. So $G$ is an assumption-free self-test for $\tilde S_{\text{res}}$.
\end{proof}

(The same conclusion for pure self-tests from correlation has been shown also in \cite[Proposition 4.14]{paddock2023operatoralgebraic} via a different approach.) Notice that $\tilde S_{\text{res}}$ is both projective and full-rank in the above proof. This means that for assumption-free self-tests, we can always take the canonical strategy $\tilde S$ to be both projective and full-rank. This is somehow also the best one can hope for, because we can never show that the canonical strategy \emph{is} projective and full-rank: consider $\tilde S'=(\ket{\tilde\psi}\otimes \ket{0}_A\ket{0}_B,\{\tilde A_{sa}\otimes \mathbbm{1}\},\{\tilde B_{tb}\otimes \mathbbm{1}\})$. Then $\tilde S'\xhookrightarrow{}\tilde S$ and $\tilde S\xhookrightarrow{}\tilde S'$. So $G$ also self-tests $\tilde S'$.

As a final remark before we go into the steps of the proof of Theorem \ref{thm:bigtheorem}, we note that most of the results in this section also apply to correlation self-tests, while one of them (namely, ones in Subsection \ref{sec:mainres}) additionally requires that the correlation is extreme in the quantum set. See Appendix \ref{app:probdist} for a detailed discussion.

\subsection{Lifting the PVM assumption}

Here we show that robust pure PVM self-test implies robust pure self-test if the canonical strategy is full-rank, with the building blocks from Section \ref{sec:lemmas}. 

\begin{theorem}[robust pure PVM implies robust pure]
    Let a game $G$ be a robust pure PVM self-test for a full-rank canonical strategy $\tilde S$. Then 
    \begin{enumerate}[(a)]
        \item $\tilde S$ is a projective strategy, and
        \item $G$ is also a robust pure self-test for $\tilde S$.
    \end{enumerate}
\label{thm:PVMtoPOVMrob}
\end{theorem}

\begin{proof}
    We first prove (a). Note that robust self-tests always imply exact self-tests by taking $\varepsilon=0$ (which causes $\delta=0$).

    For any $S$ that generates the same correlation of $\tilde S$, consider its Naimark dilation $S_{\operatorname{Naimark}}$. Since $G$ is a PVM self-test for $\tilde S$, it holds that $S_{\operatorname{Naimark}}\xhookrightarrow{}\tilde S$. By the invariance of projectiveness (Proposition \ref{prop:invariantProj}), $\tilde S$ is 0-projective, thus projective.

    Now we prove (b). For any $\varepsilon$, let $\varepsilon'=\varepsilon/5$. Since $G$ is a robust pure PVM self-test, for such $\varepsilon'$ there exist $\delta'$ such that, any $\delta'$-optimal pure projective strategy $S_{\operatorname{proj}}$ for $G$ satisfies $S_{\operatorname{proj}}\xhookrightarrow{\varepsilon'}\tilde{S}$. 

    Consider a non-projective strategy $S_{\operatorname{non-proj}}$ that is $\delta'$-optimal for $G$. Since its Naimark dilation $S_{\operatorname{Naimark}}$ is projective and $\delta'$-optimal, it holds that $S_{\operatorname{Naimark}}\xhookrightarrow{\varepsilon'}\tilde S$. Note that $\tilde S$ is assumed to be full-rank (thus support-preserving), by the invariance of support-preservingness (Proposition \ref{prop:invariantSupp}) $S_{\operatorname{Naimark}}$ is $4\varepsilon'$-support-preserving. Then by Theorem \ref{thm:NaiSupp} $S_{\operatorname{non-proj}}$ is $4\varepsilon'$-support-preserving. Then $S_{\operatorname{non-proj}}\xhookrightarrow{4\varepsilon'}S_{\operatorname{Naimark}}$ by Proposition \ref{prop:Naimarkrob}. 
    By transitivity, $S_{\operatorname{non-full}}\xhookrightarrow{\varepsilon'+4\varepsilon'=\varepsilon}\tilde{S}$. 
    
    Let $\delta=\delta'$. So we conclude that $S_{\operatorname{non-proj}}\xhookrightarrow{\varepsilon}\tilde{S}$ for any $\delta$-optimal $S_{\operatorname{non-full}}$, that is, $G$ is also a robust pure self-test.
\end{proof}

\begin{remark}~
    \begin{itemize}
    \item Previously, work \cite[Theorem 3.7]{paddock2023operatoralgebraic} shows that in some special cases where the correlation is synchronous or binary, PVM assumption can be lifted for exact self-tests. Here we show that this in fact be done in a more general scenario, and for robust self-tests as well. 
    \item Exact version of the (b) part of the theorem and its proof hold automatically by taking $\varepsilon=0$ (which causes $\delta=0$).
    \item If there is already an explicit $(\delta,\varepsilon)$ dependence in the PVM self-test, \emph{e.g.}, $\varepsilon=O(\delta^2)$, then our proof still works and give the result that any $\delta$-optimal strategy is a $5O(\delta^2)$-local-dilation.
\end{itemize}
\end{remark}
\subsection{Lifting the full-rank assumption}\label{sec:rank}

Once again, using the tools from Section \ref{sec:lemmas}, we will now show we can get rid of the full-rank assumption if our canonical strategy is projective. 

\begin{theorem}
    Let a game $G$ be a robust pure full-rank self-test for a projective canonical strategy $\tilde{S}$. Then 
    \begin{enumerate}[(a)]
        \item $\tilde S$ is support-preserving, and
        \item $G$ is also a robust pure self-test for $\tilde{S}$.
    \end{enumerate}
\label{thm:full_to_any_rank_rob}
\end{theorem}

\begin{proof}
    We first prove (a). Note that robust self-tests always imply exact self-tests by taking $\varepsilon=0$ (which causes $\delta=0$).
    
    For any $S$ that generates the same correlation of $\tilde S$, consider its restriction $S_{\operatorname{res}}$. Since $G$ is a full-rank self-test for $\tilde S$ it holds that $S_{\operatorname{res}}\xhookrightarrow{}\tilde S$. By the invariance of support-preservingness (Proposition \ref{prop:invariantSupp}), $\tilde S$ is support-preserving.

    Now we prove (b). For any $\varepsilon$, let $\varepsilon'$ be the positive number such that $\varepsilon'+\sqrt{3\varepsilon'}=\varepsilon$. Since $G$ is a robust pure full-rank self-test, for such $\varepsilon'$ there exist $\delta'$ such that, any $\delta'$-optimal pure full-rank strategy $S_{\operatorname{full}}$ for $G$ satisfies $S_{\operatorname{full}}\xhookrightarrow{\varepsilon'}\tilde{S}$. 

    Consider a non-full-rank strategy $S_{\operatorname{non-full}}$ that is $\delta'$-optimal for $G$. Since its restriction $S_{\operatorname{res}}$ is full-rank and $\delta'$-optimal, it holds that $S_{\operatorname{res}}\xhookrightarrow{\varepsilon'}\tilde S$. Note that $\tilde S$ is assumed to be projective, by the invariance of projectiveness (Proposition \ref{prop:invariantProj}) $S_{\operatorname{res}}$ is $\sqrt{3\varepsilon'}$-projective. Then by Theorem \ref{thm:resProj}, $S_{\operatorname{non-full}}$ is $\sqrt{3\varepsilon'}$-support-preserving. Then $S_{\operatorname{non-full}}\xhookrightarrow{\sqrt{3\varepsilon'}}S_{\operatorname{res}}$ by Proposition \ref{prop:restrictrob}. 
    By transitivity, $S_{\operatorname{non-full}}\xhookrightarrow{\varepsilon'+\sqrt{3\varepsilon'}=\varepsilon}\tilde{S}$. 
    
    Let $\delta=\delta'$. So we conclude that $S_{\operatorname{non-full}}\xhookrightarrow{\varepsilon}\tilde{S}$ for any $\delta$-optimal $S_{\operatorname{non-full}}$, that is, $G$ is also a robust pure self-test.
\end{proof}

\begin{remark}~
    \begin{itemize}
        \item Exact version of the (b) part of the theorem and its proof hold automatically by taking $\varepsilon=0$ (which causes $\delta=0$).
        \item If there is already an explicit $(\delta,\varepsilon)$ dependence in the full-rank self-test, \emph{e.g.}, $\varepsilon=O(\delta^2)$, then our proof still works and give the result that any $\delta$-optimal strategy is a $O(\delta)$-local-dilation.
    \end{itemize}
\end{remark}
\subsection{Lifting the purity assumption}\label{sec:mainres}

Until now, we focused on strategies using a pure state. In general, there might be strategies that use mixed states. Note that many self-testing theorems are shown only for pure strategies, therefore it is worth investigating whether or not those theorems also hold for mixed strategies. We address this task and show that a pure self-test is a mixed self-test as long as the canonical strategy has a pure, full-rank state. The main result in this subsection is the following. 

\begin{theorem}
\label{thm:pure_to_mix_rob}
Let $t\subseteq\{\text{PVM}\}$ and $G$ robust pure $t$ self-tests $\tilde{S}$, where $\ket{\tilde{\psi}}$ has full Schmidt rank. Then $G$ robust mixed $t$ self-tests $\tilde{S}$.
\end{theorem}

\subsubsection{Eigenspace of game operator}

Let $\tilde{S}=(\ket{\tilde\psi},\{\tilde{A}_{sa}\},\{\tilde{B}_{tb}\})$ be the canonical strategy self-tested by a game $G = (\mathcal S,\mathcal T,\mathcal A,\allowbreak\mathcal B, \pi, \mathcal V)$. We want to understand the set $\tilde{Q}$ of all quantum states that yield optimal quantum strategies when measured using the measurements from $\tilde{S}$. We will show that $\ket{\tilde\psi}$ is the only state that wins the game optimally. Define the operator
\begin{equation*}
	\tilde{W} := \sum_{a,b,s,t}  \pi(s,t)\mathcal{V}(a,b|s,t)(\tilde{A}_{sa}\otimes \tilde{B}_{tb}).
\end{equation*}
Then that is equivalent to that the eigenspace of the largest eigenvalue of $\tilde{W}$ is 1-dimensional. We introduce key lemmas, then present the proof of this property.

Given an optimal strategy $S=(\ket\psi,\{A_{sa}\},\{B_{tb}\})$ for $G$, the following lemma characterizes this set $Q$, the set of all quantum states that yield optimal quantum strategies, in terms of the $W$ operator. 
\begin{lemma}\label{lem:ptm_space}
	Let ${S}=(\ket\psi\in\mathcal{H}_A\otimes\mathcal{H}_B,\{A_{sa}\},\{B_{tb}\})$ be any optimal pure strategy for some game $G = (\mathcal S,\mathcal T,\mathcal A,\mathcal B, \pi, \mathcal V)$. Define the set 
	$$Q := \{\ket\phi \in \mathcal H_A \tensor \mathcal H_B : \omega((\ket\phi,\{A_{sa}\},\{B_{tb}\}),G)=\omega_q(G)\}.$$
	Define the operator
	\begin{equation*}
	    W := \sum_{a,b,s,t}  \pi(s,t)\mathcal{V}(a,b|s,t)(A_{sa}\otimes B_{tb}).
	\end{equation*}
	Then $Q = \Statespace(V_{\max(\sigma(W))})$, where $V_\lambda$ denotes the $\lambda$-eigenspace of $W$ and $\sigma(W)$ is the spectrum of $W$.
\end{lemma}
\begin{proof}
Let $\lambda_{\max} := \max\left(\sigma(W)\right)$. 
A key thing to observe is that $\bra{\psi'} W \ket{\psi'}$ is the success probability of strategy $S'=\left(\{A_{sa}\}, \{B_{tb}\}, \ket{\psi'}\right)$ (see Lemma~\ref{lem:bck_score}). It now follows that
\begin{equation} 
    \bra{\psi'} W \ket{\psi'} = \omega_q(G)         
    \quad \Longleftrightarrow \quad 
    \text{$S'$ is an optimal strategy for $G$}
    \quad \Longleftrightarrow \quad
    \ket{\psi'}\in Q.
\label{eq:stateQ}
\end{equation}
Since the unit vectors $\ket\psi$ that maximize the expression $\bra{\psi'} W \ket{\psi'}$ are precisely the quantum states in $\Statespace\left(V_{\lambda_{\max}}\right)$, the desired statement follows.   
\end{proof}

Our main use for this lemma is that it, for pure self-tested strategies, allows us to characterise $Q$ in relation to a vector space. In particular, $Q = \Statespace(\Span(Q))$.

While the above result already characterises the state space of optimal states for a given set of measurement operators, it turns out that in some cases, it is actually possible to say more; namely that if all those states have full Schmidt rank, then that space necessarily has dimension 1. The statement and proof of this result due to Cubitt, Montanaro and Winter \cite{Cubitt_2008}, though it is restated here for the sake of completeness.
\begin{lemma}[\cite{Cubitt_2008}]\label{lem:spacebound}
	Let $S$ be a subspace of the bipartite space $\mathcal{H}_A \otimes \mathcal{H}_B$, where $\dim \mathcal{H}_A = \dim \mathcal{H}_B = d$. If every nonzero state in $S$ has Schmidt rank $d$, then $\dim S = 1$. 
\end{lemma}

\begin{proof}
    For contradiction, we assume that there exists an (at least) two-dimensional subspace $S$ where every unit vector in $S$ has Schmidt rank $d$. Let $\ket\varphi ,\ket\psi \in S\subseteq\mathbb C^d \otimes \mathbb C^d$ be linearly independent. We will show that there exists $x\in\mathbb{C}$ such that the (unnormalised) $\ket{\phi_x} = \ket\varphi + x\ket\psi$ has Schmidt rank less than $d$, which contradicts the hypothesis.

	Arrange the coefficients of vectors in the computational basis $\{\ket{i}\ket{j}\}_{i,j=0,\ldots,d-1}$, into a $d\times d$ matrix. Then the Schmidt rank of the state vector equals the
	linear rank of the associated matrix. So, $\ket{\psi_x}$ has Schmidt rank less than $d$ if and only if the determinant of the associated matrix of $\ket{\psi_x}$ is $0$. Also note that the determinant is a non-constant polynomial in $x$ of degree
	$d$. Hence, it must have a root $x_0\in\mathbb{C}$, and
	the corresponding $\ket{\phi_{x_0}}$ has Schmidt rank less than $d$.
\end{proof}
Apart from this, we will need another result characterising the Schmidt rank of optimal states; namely that they all have minimal Schmidt rank. This will be useful for proving that the conditions of the above result holds.
\begin{lemma}\label{lem:minselftest}
	Let $G$ be nonlocal game that pure self-tests the strategy $\tilde S = (\ket{\tilde\psi},\{\tilde{A}_{sa}\},\{\tilde{B}_{tb}\})$. Then $\ket{\tilde\psi}$ has minimum Schmidt rank across the states of all optimal pure strategies. 
\end{lemma}
\begin{proof}
	Let $t$ be the Schmidt rank of $\ket{\tilde\psi}$. For contradiction, assume that there exists a strategy $S = (\ket{\phi},\{{A}_{sa}\},\{{B}_{tb}\})$ with Schmidt rank $s < t$. By the definition of pure self-testing, there exist local isometries $V_A,V_B$ and a state $\ket{\aux}$ such that
	\begin{equation*}
		(V_A \otimes V_B) \ket{\phi} = \ket{\tilde\psi} \otimes \ket{\aux}.
	\end{equation*}
	Since local isometries preserve Schmidt rank, the Schmidt rank (with respect to Alice/Bob partition) of $\ket{\tilde\psi} \otimes \ket{\aux}$ is $s$. This is a contradiction, since tensoring with a state $\ket{\aux}$ cannot decrease the Schmidt rank of $\ket{\tilde\psi}$.
\end{proof}

Now we are ready to present the following proposition:

\begin{proposition}
\label{prop:dim1}
    Let $\tilde{S}=(\ket{\tilde\psi},\{\tilde{A}_{sa}\},\{\tilde{B}_{tb}\})$ be the canonical full-rank strategy self-tested by a game $G = (\mathcal S,\mathcal T,\mathcal A,\mathcal B, \pi, \mathcal V)$. Define the operator
\begin{equation*}
	\tilde{W} := \sum_{a,b,s,t}  \pi(s,t)\mathcal{V}(a,b|s,t)(\tilde A_{sa}\otimes \tilde B_{tb}),
\end{equation*}
and the set
\begin{equation*}
    \tilde Q:=\{\ket{\phi}:\braket{\phi|\tilde W|\phi}=\omega_q(G)\}.
\end{equation*}
Then $\dim\Span\{\tilde Q\}=1$. 
\end{proposition}

\begin{proof}
    Let $\lambda_0$ be the largest eigenvalue of $\tilde W$, which is also the quantum value of $G$. Let $V_{\lambda_0}$ be the eigenspace of $\tilde W$ corresponding to $\lambda_0$. Then $\tilde Q$ coincides the set of all unit vectors in $V_{\lambda_0}$. 
    
    By Lemma \ref{lem:minselftest}, $\ket{\tilde\psi}$ has the minimal rank among states in $\tilde Q$. Then in $V_{\lambda_0}$ all non-zero vectors are full-rank. By Lemma \ref{lem:spacebound}, $\dim V_{\lambda_0}=1=\dim\Span(\tilde Q)$.
\end{proof}

\subsubsection{Pure self-tests imply mixed self-tests}

The following lemma can be seen as a first step in the proof of showing that a pure, robust self-test is a mixed, robust self-test, if the canonical quantum strategy has a state of full Schmidt rank. It shows that any purification of a mixed quantum strategy can be $\varepsilon'$-dilated to a quantum strategy that uses the operators of the canonical strategy of the pure self-test. 

\begin{lemma}\label{lem1robmixed}
Let $G$ be a robust, pure self-test for $\tilde S=\left(\ket{\tilde\psi},\{\tilde A_{sa}\},\{\tilde B_{tb}\}\right)$. Let $\rho_{AB}$ be a mixed state for which $S=\left(\rho_{AB},\{ A_{sa}\},\{ B_{tb}\}\right)$ is a $\delta$-optimal strategy and consider $S^{(1)}=\left(\ket{\psi}_{ABP},\{ A_{sa}\},\{ B_{tb}\otimes \mathbbm{1}_P\}\right)$, where $\ket{\psi}_{ABP}$ is a purification of $\rho_{AB}$. Then $S^{(2)}=\left(X\ket{\psi}_{ABP},\{ \tilde A_{sa}\otimes \mathbbm{1}_{\check{A}}\},\{ \tilde B_{tb}\otimes \mathbbm{1}_{\hat{B}}\otimes \mathbbm{1}_P\}\right)$ is a local $2 \varepsilon$-dilation of $S^{(1)}$, where $X$ is an isometry obtained from the robust, pure self-test. 
\end{lemma}

\begin{proof}
We have two pure strategies, $S^{(1)}$ and $\left(\ket{\psi}_{ABP},\{ A_{sa}\otimes \mathbbm{1}_P\},\{ B_{tb}\}\right)$, which are $\delta$-optimal. Then by the pure robustness, we have that\allowdisplaybreaks
\begin{align}
    &V_{AP}\otimes V_B[(A_{sa}\otimes \mathbbm{1}_P)\otimes \mathbbm{1}_B\ket{\psi}_{ABP}]\approx_{\varepsilon}(\tilde{A}_{sa}\otimes \mathbbm{1}_{\tilde B}\ket{\tilde{\psi}})\otimes\ket{\operatorname{aux}_1},\nonumber\\
     &V_{AP}\otimes V_B[(\mathbbm{1}_A\otimes \mathbbm{1}_P)\otimes B_{tb}\ket{\psi}_{ABP}]\approx_{\varepsilon}(\mathbbm{1}_{\tilde A}\otimes \tilde{B}_{tb}\ket{\tilde{\psi}})\otimes\ket{\operatorname{aux}_1},\label{eq:s1m}\\
    &V_{AP}\otimes V_B[\ket{\psi}_{ABP}]\approx_{\varepsilon}\ket{\tilde{\psi}}\otimes\ket{\operatorname{aux}_1},\label{eq:s1s}\\
    &W_{A}\otimes W_{BP}[A_{sa}\otimes (\mathbbm{1}_B\otimes \mathbbm{1}_P)\ket{\psi}_{ABP}]\approx_{\varepsilon}(\tilde{A}_{sa}\otimes \mathbbm{1}_{\tilde B}\ket{\tilde{\psi}})\otimes\ket{\operatorname{aux}_2},\label{eq:s2m}\\
    &W_{A}\otimes W_{BP}[\mathbbm{1}_A\otimes (B_{tb}\otimes \mathbbm{1}_P)\ket{\psi}_{ABP}]\approx_{\varepsilon}(\mathbbm{1}_{\tilde A}\otimes \tilde{B}_{tb}\ket{\tilde{\psi}})\otimes\ket{\operatorname{aux}_2},\nonumber\\
    &W_{A}\otimes W_{BP}[\ket{\psi}_{ABP}]\approx_{\varepsilon}\ket{\tilde{\psi}}\otimes\ket{\operatorname{aux}_2}\label{eq:s2s}.
\end{align}

Let $X:=W_A\otimes V_B \otimes \mathbbm{1}_P$. We need to show that
\begin{align*}
    X[A_{sa}\otimes \mathbbm{1}_B\otimes \mathbbm{1}_P\ket{\psi}_{ABP}]&\approx_{2\varepsilon}(\tilde{A}_{sa}\otimes \mathbbm{1}_{\tilde B}\otimes \mathbbm{1}_{\check{A} \hat{B} P})X\ket{\psi}_{ABP}, \\
    X[\mathbbm{1}_A\otimes B_{tb}\otimes \mathbbm{1}_P\ket{\psi}_{ABP}]&\approx_{2\varepsilon}(\mathbbm{1}_{\tilde A}\otimes \tilde{B}_{tb}\otimes \mathbbm{1}_{\check{A} \hat{B} P})X\ket{\psi}_{ABP}
\end{align*}
for all $a,b,s,t$.

Equations \eqref{eq:s1m}, \eqref{eq:s1s} imply 
\begin{align*}
    &V_{AP}\otimes V_B[(\mathbbm{1}_A\otimes \mathbbm{1}_P)\otimes B_{tb}\ket{\psi}_{ABP}]\approx_{2\varepsilon}(\mathbbm{1}_{\tilde{A}}\otimes \tilde{B}_{tb}\otimes \mathbbm{1}_{\hat{A} \hat{B} P})(V_{AP}\otimes V_B)[\ket{\psi}_{ABP}].
\end{align*}

Applying $V_{AP}^*\otimes \mathbbm{1}_{\tilde B \hat{B}}$ to the left of both sides yields
\begin{align}
    \mathbbm{1}_{AP}\otimes V_BB_{tb}[\ket{\psi}_{ABP}]\approx_{2\varepsilon}&\mathbbm{1}_{AP}\otimes (\tilde{B}_{tb}\otimes \mathbbm{1}_{\hat{B}})V_B[\ket{\psi}_{ABP}]\label{eq:vb}.
\end{align}
Similarly, we obtain 
\begin{align}
    &W_{A}A_{sa}\otimes \mathbbm{1}_{BP}[\ket{\psi}_{ABP}]\approx_{2\varepsilon}(\tilde{A}_{sa}\otimes \mathbbm{1}_{\check{A}})W_A\otimes \mathbbm{1}_{BP}[\ket{\psi}_{ABP}]\label{eq:wa}
\end{align}
from the equations \eqref{eq:s2m}, \eqref{eq:s2s}. 

Now, applying $W_A\otimes \mathbbm{1}_{\tilde B \hat{B} P}$ to the left of both sides of equation \eqref{eq:vb} gives us
\begin{align*}
    &W_{A}\otimes V_B\otimes \mathbbm{1}_P[\mathbbm{1}_A\otimes B_{tb}\otimes \mathbbm{1}_P\ket{\psi}_{ABP}]\approx_{2\varepsilon}(\mathbbm{1}_{\tilde A}\otimes \tilde{B}_{tb}\otimes \mathbbm{1}_{\check{A} \hat{B} P})(W_{A}\otimes V_B\otimes \mathbbm{1}_P)\ket{\psi}_{ABP}.
\end{align*}

Finally, we deduce
\begin{align*}
    W_{A}\otimes V_B\otimes \mathbbm{1}_P[A_{sa}\otimes \mathbbm{1}_B \otimes \mathbbm{1}_P\ket{\psi}_{ABP}]&\approx_{2\varepsilon}(\tilde{A}_{sa}\otimes \mathbbm{1}_{\tilde B}\otimes \mathbbm{1}_{\check{A} \hat{B} P})(W_{A}\otimes V_B\otimes \mathbbm{1}_P)\ket{\psi}_{ABP}
\end{align*}
from applying $V_B\otimes \mathbbm{1}_{\tilde A \check{A} P}$ to the left of both sides of equation \eqref{eq:wa}.
\end{proof}

The next lemma shows that the strategy we constructed before is actually almost optimal for the nonlocal game.

\begin{lemma}\label{lem2robmixed}
Let $G$ be a pure, robust self-test of $\tilde{S}$ and let $S^{(2)}$ be as in Lemma \ref{lem1robmixed}. Then $S^{(2)}$ is $(\delta+C\varepsilon)$-optimal, where $C$ depends on $G$. 
\end{lemma}

\begin{proof}
Let $p_1(a,b|s,t)$ and $p_2(a,b|st)$ be the correlation of $S^{(1)}$ and $S^{(2)}$, respectively. It holds 
\begin{align*}
    &|p_1(a,b|s,t)-p_2(a,b|s,t)|\\
    &=|\mathrm{Tr} (((A_{sa} \otimes B_{tb} \otimes \mathbbm{1}_P) \\
    &\hspace{2cm}- (W_A \otimes V_B \otimes \mathbbm{1}_P)^* (\tilde{A}_{sa} \otimes \tilde{B}_{tb} \otimes \mathbbm{1}_{\check{A}\hat{B}P})(W_A \otimes V_B \otimes \mathbbm{1}_P)) \ket{\psi}\bra{\psi}_{ABP})|\\
     &=|\langle (W_A \otimes V_B \otimes \mathbbm{1}_P) \ket{\psi}, ( (W_A \otimes V_B \otimes \mathbbm{1}_P)(A_{sa} \otimes B_{tb} \otimes \mathbbm{1}_P)\\
      &\hspace{7cm}-(\tilde{A}_{sa} \otimes \tilde{B}_{tb} \otimes \mathbbm{1}_{\check{A}\hat{B}P})(W_A \otimes V_B \otimes \mathbbm{1}_P)) \ket{\psi} \rangle |\\
     &\leq|| (W_A \otimes V_B \otimes \mathbbm{1}_P) \ket{\psi} || \cdot || ( (W_A \otimes V_B \otimes \mathbbm{1}_P)(A_{sa} \otimes B_{tb} \otimes \mathbbm{1}_P)\\ 
     &\hspace{7cm} -(\tilde{A}_{sa} \otimes \tilde{B}_{tb} \otimes \mathbbm{1}_{\check{A}\hat{B}P})(W_A \otimes V_B \otimes \mathbbm{1}_P)) \ket{\psi} ||\\
     &=|| ( (W_A \otimes V_B \otimes \mathbbm{1}_P)(A_{sa} \otimes B_{tb} \otimes \mathbbm{1}_P) - (\tilde{A}_{sa} \otimes \tilde{B}_{tb} \otimes \mathbbm{1}_{\check{A}\hat{B}P})(W_A \otimes V_B \otimes \mathbbm{1}_P)) \ket{\psi} ||
\end{align*}
for all $a,b,s,t$, where the inequality comes from the Cauchy-Schwarz inequality. Since $S^{(2)}$ is a local $2\epsilon$-dilation of $S^{(1)}$ by Lemma \ref{lem1robmixed}, we know 
\begin{align*}
    || ( (W_A \otimes V_B \otimes \mathbbm{1}_P)(A_{sa} \otimes B_{tb} \otimes \mathbbm{1}_P) - (\tilde{A}_{sa} \otimes \tilde{B}_{tb} \otimes \mathbbm{1}_{\check{A}\hat{B}P})(W_A \otimes &V_B \otimes \mathbbm{1}_P)) \ket{\psi} ||\\
    &\leq 2\mathrm{max}\{|O_A|,|O_B|\}\varepsilon .
\end{align*}
Thus, we get 
\begin{align*}
    |w(S^{(1)},G)-&w(S^{(2)},G)|\\
    &=|\sum_{s,t}\pi(s,t)\sum_{a,b}\mathcal{V}(a,b|s,t)(p_1(a,b|s,t)-p_2(a,b|s,t))|\\
    &\leq\sum_{s,t}\pi(s,t)\sum_{a,b}\mathcal{V}(a,b|s,t)|(p_1(a,b|s,t)-p_2(a,b|s,t))|\\
    &\leq 2 \left(\sum_{s,t}\pi(s,t)\sum_{a,b}\mathcal{V}(a,b|s,t)\right)\mathrm{max}\{|O_A|,|O_B|\}\varepsilon.
\end{align*}
Since $S^{(1)}$ is $\delta$-optimal, we deduce that $S^{(2)}$ is $(\delta + C\varepsilon)$-optimal.
\end{proof}

Finally, we will see that the almost optimal strategy from the previous lemma can be $\varepsilon'$-dilated to the canonical strategy of the pure, robust self-test. 

\begin{lemma}\label{lem3robmixed}
Let $G$ be a pure, robust self-test of $\tilde{S}$, where $\ket{\tilde{\psi}}$ has full Schmidt rank and let $S^{(2)}$ be as in Lemma \ref{lem1robmixed}. Then $\tilde{S}$ is a local $(\sqrt{2\frac{\delta+C\varepsilon}{\Delta}})$-dilation of $S^{(2)}$, where $\Delta$ depends on $\tilde{S}$ and $G$. 
\end{lemma}

\begin{proof}
We will show that
\begin{align*}
    \|X\ket{\psi}_{ABP}-\ket{\tilde{\psi}}\otimes\ket{\operatorname{aux}}\|\le\sqrt{2\frac{\delta+C\varepsilon}{\Delta}}.
\end{align*}

Define the game operator $\tilde W:=\sum_{a,b,s,t}\pi(s,t)\mathcal{V}(a,b|s,t)(\tilde{A}_{sa}\otimes\tilde{B}_{tb})$. Let $\{\lambda_i\}$ be the eigenvalues of $\tilde W$ (ordered decreasingly), then by Proposition \ref{prop:dim1} we have $\lambda_0>\lambda_1$.
Consider the decomposition of $X\ket{\psi}_{ABP}$ in the eigenspaces of $\tilde W$
\begin{align*}
    X\ket{\psi}_{ABP}=\sum_{j=0}^{d-1}\sqrt{p_j}\ket{\varphi_j}_{\tilde{A}\tilde{B}}\ket{\operatorname{aux}_j}_{\check{A}\hat{B}P},
\end{align*}
where $\{\ket{\varphi_j}\}$ are eigenvectors of $\tilde W$. By definition $\ket{\varphi_0}=\ket{\tilde{\psi}}$. We will now lower bound $p_0$ to see that $X\ket{\psi}_{ABP}$ is close to $\ket{\tilde{\psi}}\otimes\ket{\operatorname{aux}_0}$. 
It holds
\begin{align*}
    w(S^{(2)},G)=\tr[W\left(\sum_ip_i\proj{\varphi_i}\right)]=\sum_ip_i\lambda_i\le p_0\lambda_0+(1-p_0)\lambda_1
\end{align*}
for the game value of $S^{(2)}$. On the other hand, we know that $w(S^{(2)},G)\ge\lambda_0-(\delta+C\varepsilon)$ by Lemma \ref{lem2robmixed}. Therefore, we have
\begin{align*}
    \lambda_0-(\delta+C\varepsilon)\le p_0\lambda_0+(1-p_0)\lambda_1,
\end{align*}
which implies $p_0\ge1-\frac{\delta+C\varepsilon}{\Delta}$, where $\Delta:=\lambda_0-\lambda_1$. Then
\begin{align*}
    \|X\ket{\psi}_{ABP}-\ket{\tilde{\psi}}\otimes\ket{\operatorname{aux}_0}\|=&\sqrt{(1-\sqrt{p_0})^2+\sum_{j>0}p_j}\nonumber\\
    =&\sqrt{2-2\sqrt{p_0}}\\
    \le&\sqrt{2\frac{\delta+C\varepsilon}{\Delta}}.
\end{align*}
We therefore get 
\begin{align*}
    \|(\mathbbm{1}_{\tilde A}\otimes \tilde{B}_{tb}\otimes \mathbbm{1}_{\hat{B}P})X\ket{\psi}_{ABP}&-(\mathbbm{1}_{\tilde A}\otimes \tilde{B}_{tb})\ket{\tilde{\psi}}\otimes\ket{\operatorname{aux}_0}\|\\
    &\leq \|X\ket{\psi}_{ABP}-\ket{\tilde{\psi}}\otimes\ket{\operatorname{aux}_0}\|\\
    &\leq\sqrt{2\frac{\delta+C\varepsilon}{\Delta}},
\end{align*}
since $\mathbbm{1}_{\tilde A}\otimes \tilde{B}_{tb}\otimes \mathbbm{1}_{\hat{B}P}$ is a contraction. Similarly, we obtain
\begin{align*}
 \|(\tilde{A}_{sa}\otimes \mathbbm{1}_{\tilde B}\otimes \mathbbm{1}_{\hat{B}P})X\ket{\psi}_{ABP}&-(\tilde{A}_{sa}\otimes \mathbbm{1}_{\tilde B})\ket{\tilde{\psi}}\otimes\ket{\operatorname{aux}_0}\|\leq\sqrt{2\frac{\delta+C\varepsilon}{\Delta}}.
\end{align*}
This finishes the proof. 
\end{proof}

By putting together the previous lemmas, we can prove Theorem \ref{thm:pure_to_mix_rob}.

\begin{proof}[Proof of Theorem \ref{thm:pure_to_mix_rob}]
 Let $\varepsilon\geq0$ and let $S$ be a $\delta$-optimal, mixed strategy, where we choose $\delta$ as in the robust pure self-test. Then by Lemmas \ref{lem1robmixed} and \ref{lem3robmixed} as well as transitivity, we know that $\tilde{S}$ is a local $(2\varepsilon+ \sqrt{2\frac{\delta+C\varepsilon}{\Delta}})$-dilation of the pure quantum strategy $S^{(1)}$ associated to $S$.
\end{proof}

We note that, Theorem \ref{thm:pure_to_mix_rob}, or more specifically, Lemma \ref{lem3robmixed} does not directly translate to self-testing from correlation, essentially because that there is no game operator for self-testing from correlation. Nevertheless, the result still holds if we impose an additional requirement of the correlation to be extreme. See Appendix \ref{app:probdist} for a full proof.
\subsection{Proof of Theorem \ref{thm:bigtheorem}}

We are now ready to prove our main theorem:

\begin{proof}[Proof of Theorem \ref{thm:bigtheorem}]
    (a): by Theorem \ref{thm:PVMtoPOVMrob}, $G$ is a robust pure self-test for $\tilde S$. Then by Theorem \ref{thm:pure_to_mix_rob} $G$ is an assumption-free self-test.

    (b): by Theorem \ref{thm:full_to_any_rank_rob}, $G$ is a robust pure self-test for $\tilde S$. By Theorem \ref{thm:originalpartC}, $\tilde S$ is support-preserving. So we take its restriction $\tilde S_{\text{res}}$, and $G$ also robust pure self-tests $\tilde S_{\text{res}}$. Then using Theorem \ref{thm:pure_to_mix_rob} $G$ is an assumption-free self-test for $\tilde S_{\text{res}}$. From Proposition \ref{prop:restrictrob}, $G$ is an assumption-free self-test for $\tilde S$.
\end{proof}
    \section{Equivalence of Definitions}\label{sect:equivalence}

In this section, we examine two commonly cited definitions of local dilation in existing literature. We demonstrate that under specific circumstances, each of these definitions is equivalent to with the ones we have adopted.

\subsection{Local dilation in a matrix form}

When the arbitrary strategy is mixed, rather than considering local dilation condition (Definition \ref{def:localdilation}) in a ``vector form'' one could instead consider a ``matrix-form'' condition (see Appendix~C in \cite{GeoQCor}):
\begin{definition}[Local dilation (alternative)]\label{def:dilation_mixed}
Given two strategies
    \begin{align*}
        S&=(\rho_{AB}\in B({\mathcal H}_A \tensor {\mathcal H}_B),\{ A_{sa}\}_{s\in\mathcal{S},a\in\mathcal{A}},\{ B_{tb}\}_{t\in\mathcal{T},b\in\mathcal{B}}) 
        \text{ and}
        \\
        \tilde S&=(\ket{\tilde \psi} \in {\mathcal H}_{\tilde{A}} \tensor {\mathcal H}_{\tilde{B}},\{\tilde A_{sa}\}_{s\in\mathcal{S},a\in\mathcal{A}},\{\tilde B_{tb}\}_{t\in\mathcal{T},b\in\mathcal{B}})
    \end{align*}
we write $S\xhookrightarrow{}_1 \tilde S$ if there exist spaces $\mathcal H_{\hat A},\mathcal H_{\hat B},$ a local isometry $U = U_A \tensor U_B$, with $U_A : \mathcal H_A \to \mathcal H_{\tilde A} \tensor \mathcal H_{\hat A}$, $U_B : \mathcal H_B \to \mathcal H_{\tilde B} \tensor \mathcal H_{\hat B}$ and a state $\ket{\aux} \in B(\mathcal H_{\hat A}\tensor\mathcal H_{\hat B})$ such that for all $s,t,a,b$ we have
   \begin{equation}
        \label{eq:localdilationAlt}
        U(A_{sa} \tensor B_{tb})\rho_{AB}U^* = (\tilde A_{sa} \tensor \tilde B_{tb})\proj{\tilde\psi} \tensor \sigma_\aux.
    \end{equation}
\end{definition}
This kind of definition of local dilation can also be used to construct an alternative definition for self-testing. We will in this appendix prove that the two definitions for local dilations in fact are equivalent. 
\begin{lemma}\label{lem:asym_to_pure}
    Let $\ket\phi \in \mathcal{H}_S$ and $\ket\psi \in \mathcal{H}_T$ be states and let $\{S_i\}_i\subset B(\mathcal{H}_S)$ and $\{T_i\}_i \subseteq B(\mathcal{H}_T)$ be POVMs. Let $U: \mathcal{H}_S \to \mathcal{H}_T$ be an isometry. If
    \begin{equation}\label{eq:asym_isometry}
        US_i\proj{\phi}U^* = (T_i\proj{\psi})\tensor \proj{\aux}
    \end{equation}
    for all $i$, then there exists a state $\ket{\aux'}$, such that
    \begin{equation}\label{eq:asym_to_pure_eq_pure}
        US_i\ket{\phi} = (T_i \ket{\psi})\ket{\aux'}
    \end{equation}
    and 
    \begin{equation}\label{eq:asym_to_pure_eq_mixed}
        \proj{\aux} = \proj{\aux'}
    \end{equation}
\end{lemma}
\begin{proof}
    Sum \pref{eq:asym_isometry} over $i$ to get 
    \begin{equation}\label{eq:sym_isometry}
        U\proj{\phi}U^* = \proj{\psi}\tensor\proj{\aux}
    \end{equation}
    using the fact that $\{S_i\}_i$ and $\{T_i\}_i$ are POVMs. Both of these operators have rank one, and therefore have a single non-zero eigenvalue. Observe that the eigenspace with the non-zero eigenvalue of the left-hand side is spanned by $U\ket\phi$ and the eigenspace with the non-zero eigenvalue of the right-hand side is spanned by $\ket{\psi}\ket\aux$. Furthermore, these two spaces are equivalent. We can therefore conclude that 
    \begin{equation}\label{eq:auxdef}
        U\ket{\phi} = e^{-i\gamma}\ket{\psi}\ket\aux,
    \end{equation}
    for some $\gamma \in \mathbb{R}$. They are in other words they are equal up to global phase. This phase change can simply be absorbed into $\ket\aux$, creating a new state $\ket{\aux'} = e^{-i\gamma}\ket\aux$. Right multiply \pref{eq:asym_isometry} with \pref{eq:auxdef} and we get the desired equation of
    \begin{equation}\label{eq:pure_asym_measure_isometry}
        US_i\ket{\phi} = (T_i\ket{\psi})\ket{\aux'},
    \end{equation}
    proving \pref{eq:asym_to_pure_eq_pure}. Finally, \pref{eq:asym_to_pure_eq_mixed} follows directly from the definition of $\ket{\aux'}$.
\end{proof}

We are also going to use the following observation, proven in \cite{GeoQCor} as Observation C.1:
\begin{lemma}[\cite{GeoQCor}]\label{lem:sum_to_pure}
Let $\rho^0_{ST},\rho^1_{ST} \in B(\mathcal{H}_S \tensor \mathcal{H}_T)$ be positive semidefinite operators. If
$$\rho^0_{ST} + \rho^1_{ST} = \proj{\psi}_S \tensor \sigma_T$$
for some $\proj{\psi}_S\in B(\mathcal{H}_S)$ and $\sigma_T \in B(\mathcal{H}_T)$, then 
$$\rho^i_{ST} = \proj{\psi}_S \tensor \sigma^i_T$$
for $i\in\{0,1\}$ and some $\sigma^i_T \in B(\mathcal{H}_T)$
\end{lemma}

The next theorem shows that the ``matrix-form'' condition~\eqref{eq:localdilationAlt} is in fact equivalent to the vector-form condition~\eqref{eq:localdilation}. The take-away here is that it does not matter which of the two variants of local dilation we base our self-testing definition on.

\begin{theorem}
Let $S$ and $\tilde S$ be two strategies. Then
    \[
    S \xhookrightarrow{} \tilde S 
    \quad \Longleftrightarrow \quad
    S \xhookrightarrow{}_1 \tilde S 
    \]
\end{theorem}
\begin{proof}
    Let $\tilde S=(\ket{\tilde \psi},\{\tilde A_{sa}\},\{\tilde B_{tb}\})$ and $S=(\rho_{AB},\{A_{sa}\},\{B_{tb}\})$.

    We start by showing $S \xhookrightarrow{} \tilde S \Rightarrow S \xhookrightarrow{}_1 \tilde S$. $S \xhookrightarrow{} \tilde S$ implies that for any purification $\ket{\psi}_{ABP}$ of $\rho_{AB}$ there exists local isometries $U = U_A\tensor U_B$ and a state $\ket{\aux}$ such that
    \begin{equation}\label{eq:mixed_sb_to_our_selftest_measurement}
        (U\otimes\mathbbm{1}_P)(A_{sa} \tensor B_{tb}\otimes \mathbbm{1}_P)\ket{\psi}_{ABP} = (\tilde A_{sa} \tensor \tilde B_{tb})\ket{\tilde\psi}_{\tilde A \tilde B} \tensor \ket{\aux}_{\hat A \hat B}.
    \end{equation}
    If we sum over $a,b$, using $\{A_{sa}\otimes B_{tb}\}_{a,b}$ is a POVM, we get
    \begin{equation}\label{eq:mixed_sb_to_our_selftest_pure}
        (U\otimes\mathbbm{1}_P)\ket{\psi}_{ABP} = \ket{\tilde\psi}_{\tilde A \tilde B} \tensor \ket{\aux}_{\hat A \hat B}.
    \end{equation}
    We can now combine \pref{eq:mixed_sb_to_our_selftest_measurement} and \pref{eq:mixed_sb_to_our_selftest_pure} through an outer product to get
    \begin{equation*}
        (U\tensor \mathbbm{1})(A_{sa} \tensor B_{tb} \tensor \mathbbm{1}_P)\proj{\psi}_{ABP}(U\tensor \mathbbm{1}_P)^* = (\tilde A_{sa} \tensor \tilde B_{tb})\proj{\tilde\psi}_{\tilde A \tilde B} \tensor \proj{\aux}_{\hat A \hat B P}.
    \end{equation*}
    Finally, by using that $\proj{\psi}_{ABP}$ is a purification of $\rho_{AB}$ and tracing out the purification space, we get 
    \begin{equation*}
        U(A_{sa} \tensor B_{tb})\rho_{AB}U^* = (\tilde A_{sa} \tensor \tilde B_{tb})\proj{\tilde\psi}_{\tilde A \tilde B} \tensor \tr_P(\proj{\aux}_{\hat A \hat B P}).
    \end{equation*}
    which shows $S \xhookrightarrow{}_1 \tilde S$.
    
    Secondly, we will show $S \xhookrightarrow{}_1 \tilde S \Rightarrow S \xhookrightarrow{} \tilde S$. $S \xhookrightarrow{}_1 \tilde S$ implies there exist a local isometry $U=U_A\tensor U_B$ such that
    \begin{equation}\label{eq:mixed_sb_to_our_selftest_mixed_condition}
        U(A_{sa} \tensor B_{tb})\rho_{AB}U^* = (\tilde A_{sa} \tensor \tilde B_{tb})\proj{\tilde\psi} \tensor \sigma_\aux.
    \end{equation}
    for some state $\sigma_\aux \in B(\mathcal{H}_{\hat A}\tensor \mathcal{H}_{\hat B})$. If we sum both sides over $a,b$, using $\{A_{sa}\otimes B_{tb}\}_{a,b}$ is a POVM, then we get 
    \begin{equation}\label{eq:mixed_sb_to_our_selftest_no_op}
        U\rho_{AB}U^* = \proj{\tilde\psi} \tensor \sigma_\aux.
    \end{equation}
    Now, consider any purification $\ket{\psi}_{ABP}$ of $\rho_{AB}$. Consider the Schmidt decomposition of $\ket{\psi}_{ABP}$ over the spaces $(\mathcal{H}_A\tensor \mathcal{H}_B)$ and $\mathcal{H}_P$. This gives
    \begin{equation}\label{eq:mixed_sb_to_our_schmidt}
        \ket{\psi}_{ABP} = \sum_{i}\lambda_i \ket{\alpha_i}_{AB}\ket{\beta_i}_{P}.
    \end{equation}
    If we trace out the purification space of this state, using that it indeed is a purification, we get that
    \begin{equation}\label{eq:mixed_sb_to_our_selftest_schmidt}
        \rho_{AB} = \sum_{i} \lambda_i^2 \proj{\alpha_i}.
    \end{equation}
    If we now substitute \pref{eq:mixed_sb_to_our_selftest_schmidt} into \pref{eq:mixed_sb_to_our_selftest_no_op2} we arrive at
    \begin{equation}\label{eq:mixed_sb_to_our_selftest_no_op2}
        \sum_{i}\lambda_i^2 U\proj{\alpha_i}U^* = \proj{\tilde\psi} \tensor \sigma_\aux.
    \end{equation}
    Observe that this is a sum of positive semidefinite operators, and therefore by \pref{lem:sum_to_pure} we have 
    \begin{equation}\label{eq:auxidef}
        U\proj{\alpha_i}U^* = \proj{\tilde\psi} \tensor \proj{\aux_i},
    \end{equation}
    for some pure state $\ket{\aux_i}\in\mathcal{H}_A\tensor\mathcal{H}_B$. By \pref{lem:asym_to_pure}, there exists a state $\ket{\aux'_i}$ such that
    \begin{equation}\label{eq:mixed_sb_to_our_pure}
        U\ket{\alpha_i} = \ket{\tilde\psi}\ket{\aux_i'}.
    \end{equation}
    and $\proj{\aux'_i}=\proj{\aux_i}$.
    We now fix $i$ and right-multiply \pref{eq:mixed_sb_to_our_selftest_mixed_condition} with \pref{eq:mixed_sb_to_our_pure}, after simplifying both sides, this gives
    \begin{equation}\label{eq:mixed_sb_to_our_schmidt_applied}
        U(A_{sa} \tensor B_{tb})\ket{\alpha_i} = \left((\tilde A_{sa} \tensor \tilde B_{tb})\ket{\tilde\psi}\right) \ket{\aux_i'}
    \end{equation}
    where we have used the fact that $\ket{\alpha_i}$ and $\ket{\alpha_j}$ are orthogonal and $\ket{\aux_i'}$ and $\ket{\aux_j'}$ are orthogonal when $i\neq j$. Finally, apply $U\tensor\mathbbm{1}$ to $\ket{\psi}_{ABP}$ as
    \begin{align*}
        (U\tensor\mathbbm{1}) (A_{sa}\tensor B_{tb}\tensor \mathbbm{1}_P)\ket{\psi}_{ABP} &= \sum_{i}\lambda_i U(A_{sa}\tensor B_{tb})\ket{\alpha_i}_{AB}\ket{\beta_i}_{P} \\
        &=  (\tilde A_{sa} \tensor \tilde B_{tb})\ket{\tilde\psi} \tensor \left(\sum_{i}\lambda_i\ket{\aux_i'}\ket{\beta_i}_{P}\right)
    \end{align*}
    where the first equality used the Schmidt decomposition of $\ket{\psi}_{ABP}$ and the second equality substituted in \pref{eq:mixed_sb_to_our_schmidt_applied}. Setting $$\ket{\aux'}=\left(\sum_{i}\lambda_i\ket{\aux_i'}\ket{\beta_i}_{P}\right)$$ 
    implies $S \xhookrightarrow{} \tilde S$.
\end{proof}

\subsection{Extraction local dilation}
\label{sec:extraction}
Here we look at a slightly different definition for local dilation between pure full-rank strategies. The idea is that we can map the POVMs from one to the other via the conjugation of unitaries.

\begin{definition}[Extraction local dilation]\label{def:dilation_measurement}
    Given two pure full-rank strategies
    \begin{align*}
        S&=(\ket{\psi}\in {\mathcal H}_A \tensor {\mathcal H}_B,\{ A_{sa}\}_{s\in\mathcal{S},a\in\mathcal{A}},\{ B_{tb}\}_{t\in\mathcal{T},b\in\mathcal{B}}) 
        \text{ and}
        \\
        \tilde S&=(\ket{\tilde \psi} \in {\mathcal H}_{\tilde{A}} \tensor {\mathcal H}_{\tilde{B}},\{\tilde A_{sa}\}_{s\in\mathcal{S},a\in\mathcal{A}},\{\tilde B_{tb}\}_{t\in\mathcal{T},b\in\mathcal{B}})
    \end{align*}
    we write $S\xhookrightarrow{}_2 \tilde S$ if there exist a local unitary $U = U_A \tensor U_B$, with $U_A : \mathcal H_A \to \mathcal H_{\tilde A} \tensor \mathcal H_{\hat A}$, $U_B : \mathcal H_B \to \mathcal H_{\tilde B} \tensor \mathcal H_{\hat B}$ such that for all $s,t,a,b$ we have
    \begin{align}
        U\ket{\psi} &= \ket{\tilde\psi} \tensor \ket{\aux},\label{eq:dilation_measurement_state}\\
        U_A A_{sa}U_A^* &= \tilde A_{sa}\otimes \mathbbm{1}_{\hat A},\label{eq:dilation_measurement_A}\\
        U_B B_{tb} U_B^* &=  \tilde B_{tb}\otimes \mathbbm{1}_{\hat B}.\label{eq:dilation_measurement_B}
    \end{align}
In case we want to name the local isometry and the auxiliary state from \eqref{eq:dilation_measurement_state}, we write $S \xhookrightarrow[U,\ket{\aux}]{}_2 \tilde{S}$.
\end{definition}

We show that for full rank strategies, this type of local dilations is in fact equivalent to the one presented in \pref{def:localdilation}. One thing that is important to note in the following lemma is that $S\xhookrightarrow{}_2 \tilde S$ is only defined when both $S$ and $\tilde S$ are pure full-rank strategies.

\begin{lemma}
Let $S$ and $\tilde S$ be two pure full-rank strategies. Then
    \[
    S \xhookrightarrow{} \tilde S 
    \quad \Longleftrightarrow \quad
    S \xhookrightarrow{}_2 \tilde S 
    \]
\end{lemma}
\begin{proof}
    Let $\tilde S=(\ket{\tilde \psi},\{\tilde A_{sa}\},\{\tilde B_{tb}\})$ and $S=(\ket{\psi},\{A_{sa}\},\{B_{tb}\})$.

    We start by showing $S \xhookrightarrow{U,\ket\aux}_2 \tilde S \Rightarrow S \xhookrightarrow{} \tilde S$. Tensor \eqref{eq:dilation_measurement_A} and \eqref{eq:dilation_measurement_B}, and right-multiply with \eqref{eq:dilation_measurement_state} to get
    $$(U_A A_{sa}U_A^*\tensor U_B B_{tb} U_B^*)U\ket{\psi} = (\tilde A_{sa}\otimes \mathbbm{1}_{\hat A} \tensor \tilde B_{tb}\otimes \mathbbm{1}_{\hat B})\ket{\tilde\psi} \tensor \ket{\aux}.$$
    By $U$ being a unitary, implying $U^*U = \mathbbm{1}$, we have 
    $$U (A_{sa}\tensor B_{tb})\ket{\psi} = \left[(\tilde A_{sa} \tensor \tilde B_{tb})\ket{\tilde\psi}\right] \tensor \ket{\aux}.$$
    implying $S \xhookrightarrow{} \tilde S$.

    We then show  $S \xhookrightarrow{V,\ket{\aux}} \tilde S \Rightarrow S \xhookrightarrow{}_2 \tilde S$. Consider the Schmidt decomposition
    \begin{equation}
        \ket{\aux}_{\hat A\hat B} = \sum_{i=0}^{r-1}\lambda_i \ket{\alpha_i} \ket{\beta_i}
    \end{equation}
    where $r$ is the Schmidt rank of $\ket{\aux}_{\hat A\hat B}$. Furthermore, observe that $\dim(\mathcal{H}_{A}) = r\cdot\dim(\mathcal{H}_{\tilde A})$. Define the isometries 
    $$T_{\hat A} := \sum_{i=0}^{r-1} \ket{\alpha_i}\!\bra{i},\qquad T_{\hat B} := \sum_{i=0}^{r-1} \ket{\beta_i}\!\bra{i}$$
    where $\ket{i}\in\mathbb{C}^r$, and observe that $T_{\hat A}T_{\hat A}^*$ is a projection onto $\Supp_{\hat A}(\ket{\aux}_{\hat A\hat B})$ and $T_{\hat B}T_{\hat B}^*$ is a projection onto $\Supp_{\hat B}(\ket{\aux}_{\hat A\hat B})$.

    Next note that
    $$\ket{\tilde\psi}\ket{\aux}=VV^*V\ket{\psi}=VV^*\ket{\tilde\psi}\ket{\aux}$$
    and so $V_AV_A^*$ act with identity on $\Supp_{\tilde A\hat A}(\ket{\tilde \psi}\ket{\aux}) = \mathcal{H}_{\tilde A} \tensor \Supp_{\hat A}(\ket{\aux})$, and similarly for $V_BV_B^*$ on $ \mathcal{H}_{\tilde B}$. Define 
    $$W_A := (\mathbbm{1}_{\tilde A}\tensor T_{\hat A}^*)V_A,\qquad W_B := (\mathbbm{1}_{\tilde B}\tensor T_{\hat B}^*)V_B, \qquad \ket{\aux'} :=(T_{\hat A}^*\tensor T_{\hat B}^*)\ket{\aux}. $$
    Observe that $\ket{\aux'} \in \mathbb{C}^r\tensor\mathbb{C}^r$ has full Schmidt rank and that $W_A$ and $W_B$ are square matrices. We claim that $W_A$ is unitary. This can be seen by
    $$W_AW_A^* =  (\mathbbm{1}_{\tilde A}\tensor T_{\hat A}^*)V_AV_A^* (\mathbbm{1}_{\tilde A}\tensor T_{\hat A}) = (\mathbbm{1}_{\tilde A}\tensor T_{\hat A}^*)(\mathbbm{1}_{\tilde A}\tensor T_{\hat A}) = \mathbbm{1}_{\tilde A}\tensor \mathbbm{1}_r$$
    using that $V_AV_A^*$ act with identity on $\mathcal{H}_{\tilde A} \tensor \Supp_{\hat A}(\ket{\aux})$ and that $T_{\hat A}$ is an isometry. A similar argument shows that $W_B$ is unitary. Now we have 
    $$(W_A\tensor W_B)(A_{sa}\tensor B_{tb})\ket{\psi} = (\mathbbm{1}_{\tilde A\tilde B}\tensor T_{\hat A \hat B})\left[(\tilde A_{sa} \tensor \tilde B_{tb})\ket{\tilde\psi}\right] \tensor \ket{\aux} = \left[(\tilde A_{sa} \tensor \tilde B_{tb})\ket{\tilde\psi}\right] \tensor \ket{\aux'}.$$
    Hence 
    $$(W_AA_{sa}W_A^*\tensor W_BB_{tb}W_B^*)(W_A\tensor W_B)\ket{\psi} = \left[(\tilde A_{sa} \tensor \tilde B_{tb})\ket{\tilde\psi}\right] \tensor \ket{\aux'}.$$
    Since $\ket{\tilde\psi}\tensor \ket{\aux'}$ has full Schmidt rank, it follows that 
    $$W_AA_{sa}W_A^* = \tilde A_{sa} \tensor \mathbbm{1}_{r},\qquad 
    W_BB_{tb}W_B^* = \tilde B_{tb}\tensor \mathbbm{1}_{r}.$$
    From this, we can conclude $S \xhookrightarrow{}_2 \tilde S$.
\end{proof}
    \section{Separation of Definitions}\label{sect:counterexamples}

\subsection{Separating pure full-rank self-tests and pure self-tests}

In Section \ref{sec:extraction} we showed that the extraction definition of local dilation and the definition we adopt are equivalent for full-rank strategies. Then it is clear that, pure full-rank ``extraction'' self-tests and pure full-rank self-tests are equivalent. Here we however shows that, there exist Bell inequalities that pure full-rank self-tests but does not pure self-test a full-rank strategy $\tilde S$. Note that this specific $\tilde S$ is not projective, therefore does not violate Theorem \ref{thm:full_to_any_rank_rob}, and moreover, it showcases the necessity of the projectiveness of $\tilde S$ as 
an assumption of Theorem \ref{thm:full_to_any_rank_rob}.

Consider the canonical strategy of CHSH game $\tilde S_{\text{CHSH}}=(\ket{\Phi^+},\{\mathcal{X},\mathcal{Z}\},\{\mathcal{H},\mathcal{G}\})$, where 
$\ket{\Phi^+}=(\ket{00}+\ket{11})/\sqrt{2}$, and $\mathcal{X},\mathcal{Z},\mathcal{H},\mathcal{G}$ are the measurements corresponding to the binary observables $X,Z,H:=\frac{1}{\sqrt{2}}(X+Z),G:=\frac{1}{\sqrt{2}}(X-Z)$, respectively. (That is, $\mathcal{X}=\{\proj{+},\proj{-}\}$, $\mathcal{Z}=\{\proj{0},\proj{1}\}$, etc.) It is well-known that the CHSH game is an assumption-free self-test for $\tilde S_{\text{CHSH}}$ \cite{McKague_2012}.

Then we incorporate a three-output POVM $\mathcal{M}=\{M_0,M_1,M_2\}$ to Bob's side, where
\begin{align*}
M_{0} &= \frac{1}{3} ( \mathbbm{1} + Z ),\\
M_{1} &= \frac{1}{3} ( \mathbbm{1} - \frac{1}{2} Z + \frac{ \sqrt{3} }{2} X ),\\
M_{2} &= \frac{1}{3} ( \mathbbm{1} - \frac{1}{2} Z - \frac{ \sqrt{3} }{2} X ).
\end{align*}

It is clear that $M_i\ge0,\sum_iM_i=\mathbbm{1}$, so $\mathcal{M}$ is a valid (non-projective) POVM. We will show that the strategy $\tilde S=(\ket{\Phi^+},\{\mathcal{X},\mathcal{Z}\},\{\mathcal{H},\mathcal{G},\mathcal{M}\})$ is full-rank self-tested. For this we need Holder's inequality 
\begin{align*}
    \tr[AB]\le\|A\|_\infty\|B\|_1,
\end{align*}
where $\|A\|_\infty:=\sup_{\|v\|=1}\|Av\|$ is the infinity norm, and $\|B\|_1:=\tr|B|=\tr[\sqrt{B^*B}]$ is the trace norm.

\begin{proposition}
Consider pure non-projective strategy $\tilde S=(\ket{\Phi^+},\{\mathcal{X},\mathcal{Z}\},\{\mathcal{H},\mathcal{G},\mathcal{M}\})$. 
\begin{enumerate}[(a)]
    \item The correlation $\tilde p$ generated by $\tilde S$ is extreme (in the quantum set of correlation).
    \item $\tilde p$ pure full-rank self-tests $\tilde S$.
\end{enumerate}
\label{prop:POVM-fullrank-selftest}
\end{proposition}

\begin{proof}
Consider a pure full-rank strategy $S=(\ket{\psi},\{\mathcal{A}_0,\mathcal{A}_1\},\{\mathcal{B}_0,\mathcal{B}_1,\mathcal{B}_2\})$ that generates $\tilde p$. Let $\mathcal{A}_i=\{A_i^+,A_i^-\}$, $\mathcal{B}_i=\{B_i^+,B_i^-\}$ for $i=0,1$ where $A_i^+,A_i^-,B_i^+,B_i^-$ are POVM elements. Define observables $A_i:=A_i^+-A_i^-$ and $B_i:=B_i^+-B_i^-$. Let $\mathcal{B}_2=\{F_0,F_1,F_2\}$. Define the following two functionals:
\begin{align*}
\beta_0 &:= \braket{\psi| A_{0}\otimes B_{0} +  A_{0}\otimes  B_{1} + A_{1}\otimes  B_{0} - A_{1} \otimes B_{1} |\psi},\\
\beta_1 &:= \braket{\psi| A_{0}\otimes  F_{0} - \frac{1}{2}  A_{0}\otimes  F_{1} +\frac{\sqrt{3}}{2} A_{1} \otimes F_{1}  - \frac{1}{2} A_{0}\otimes  F_{2}- \frac{\sqrt{3}}{2} A_{1}\otimes  F_{2}|\psi }.
\end{align*}
And, from direct calculation, one can see that $\tilde S$ satisfies $\beta_0=2\sqrt{2}$ and $\beta_1=1$. 

To prove (a), we will show that the correlation satisfying $\beta_0=2\sqrt{2}$ and $\beta_1=1$ is unique in the quantum set. 

Since the CHSH inequality is a full-rank self-test, achieving $\beta_0 = 2 \sqrt{2}$ implies that there exist unitary $U_A,U_B$ such that
\begin{align*}
U_AA_{0}U_A^* &= Z \otimes \mathbbm{1}_{A'},\\
U_AA_{1}U_A^* &= X \otimes \mathbbm{1}_{A'},\\
U_A\otimes U_B\ket{\psi} &= \ket{\Phi^{+}}_{AB} \otimes \ket{\aux}_{A'B'}.
\end{align*}
Let us now consider the three-outcome measurement and define operators:
\begin{align*}
G_{j} := \tr_{B'} \big[ (\mathbbm{1}_{B} \otimes \sigma_{B'}^{1/2})U_B^*F_{j}U_B (\mathbbm{1}_{B} \otimes \sigma_{B'}^{1/2} )\big],
\end{align*}
where $\sigma_{B'}=\tr_{A'}[\ket{\aux}\!\bra{\aux}]$. It is easy to see that the effective operators $G_{j}$ fully determine the correlation, since the observables of Alice completely ignore the $A'$ system.

Let us also define $\{ T_{j} \}_{j = 0}^{2}$ and note that they can be computed explicitly:
\begin{align*}
T_{0} &:= \tr_{AA'B'} \big[ U_A^*A_{0}U_A \otimes \mathbbm{1}_{BB'} \proj{\psi} \big] = \frac{1}{2} Z,\\
T_{1} &:= \tr_{AA'B'} \Big[ U_A^*\big( - \frac{1}{2} A_{0} + \frac{\sqrt{3}}{2} A_{1} \big)U_A \otimes \mathbbm{1}_{BB'} \proj{\psi} \Big] = \frac{1}{2} \Big( - \frac{1}{2} Z + \frac{\sqrt{3}}{2} X \Big),\\
T_{2} &:= \tr_{AA'B'} \Big[ U_A^*\big( - \frac{1}{2} A_{0} - \frac{\sqrt{3}}{2} A_{1} \big)U_A \otimes \mathbbm{1}_{BB'} \proj{\psi} \Big] = \frac{1}{2} \Big( - \frac{1}{2} Z - \frac{\sqrt{3}}{2} X \Big).
\end{align*}
It is easy to verify that the functional $\beta_1$ can be rewritten as:
\begin{align*}
\beta_1 = \sum_{j} \tr (T_{j} G_{j}).
\end{align*}
Each term can be upper-bounded using Holder's inequality:
\begin{align*}
\beta_1 \leq \sum_{j} \norm{ T_{j} }_{\infty} \norm{ G_{j} }_{1} = \frac{1}{2} \sum_{j} \tr G_{j} = 1,
\end{align*}
where we used the fact that $\norm{T_{j}}_{\infty} = \frac{1}{2}$. It is easy to determine the conditions under which these inequalities hold as equalities: since for every $T_{j}$ the positive part is one-dimensional, the $G_{j}$ operator must be proportional to these rank-1 projectors. The completeness condition allows us to deduce the proportionality constants, and finally we conclude that:
\begin{align*}
G_{0} &= \frac{1}{3} ( \mathbbm{1} + Z )=M_0,\\
G_{1} &= \frac{1}{3} ( \mathbbm{1} - \frac{1}{2} Z + \frac{ \sqrt{3} }{2} X )=M_1,\\
G_{2} &= \frac{1}{3} ( \mathbbm{1} - \frac{1}{2} Z - \frac{ \sqrt{3} }{2} X )=M_2.
\end{align*}
This allows us to fully compute the statistics, which means that it is indeed the unique correlation satisfying $\beta_0=2\sqrt{2}$ and $\beta_1=1$. Therefore, this point is an exposed point of the $\beta_0 = 2 \sqrt{2}$ face of the quantum set, and it must be (at least) extreme within the entire quantum set.

To prove (b), consider
\begin{align*}
H_{j} := ( \mathbbm{1} \otimes \sigma_{B'}^{1/2} ) U_B^*F_{j}U_B ( \mathbbm{1} \otimes \sigma_{B'}^{1/2} )
\end{align*}
and note that $G_{j} = \tr_{B'} H_{j}$. Since $G_{j}$ are rank-1 PSD operators, we must have
\begin{align*}
H_{j} = G_{j} \otimes K_{j},
\end{align*}
for some $K_{j} \geq 0$ satisfying $\tr K_{j} = 1$. Now, if $\sigma_{B'}$ is full-rank we can actually reconstruct the original measurement operators:
\begin{align*}
F_{j} = G_{j} \otimes ( \sigma_{B'}^{-1/2} K_{j} \sigma_{B'}^{-1/2} ).
\end{align*}
Using the completeness relation $\sum_{j} F_{j} = \mathbbm{1}$ and the fact that the $G_{j}$ operators correspond to an extremal three-outcome measurement on a qubit, we find that the only solution is $K_{j} = \sigma_{B'}$. Then $F_j=U_B^*(M_j\otimes\mathbbm{1}_{B'})U_B$. So $\tilde S$ is a full-rank self-tested.
\end{proof}

On the other hand, from Theorem \ref{thm:bigtheorem} (part (c)), $\tilde S$  cannot be assumption-free self-tested since it is not 0-projective. So we conclude that $\tilde S$ gives an example where a pure full-rank self-test does not imply a pure self-test.

\begin{corollary}\label{cor:fullrankselftestnoselftest}
There exists a correlation that is a pure full-rank self-test, but not a pure self-test.
\end{corollary}

Interestingly, this result provides the first instance where a quantum correlation has no pure full-rank PVM realization.

\subsection{Separating pure PVM self-tests and pure self-tests}

Here, we show that pure PVM self-tests do not necessarily imply pure self-tests with a non-full-rank canonical $\tilde S$. Specifically, we will employ the correlation given in Proposition \ref{prop:POVM-fullrank-selftest}. Also, we will need the concept of \emph{minimal} Naimark dilation \cite{Beneduci_2020}. A Naimark dilation $\{P_i\in\mathcal{B}(\mathcal{H}')\}_{i=1}^{m}$ of POVM $\{R_i\in\mathcal{B}(\mathcal{H})\}_{i=1}^{m}$ is minimal if and only if $\mathcal{H}'=\operatorname{span}\{P_{i}V\ket{\psi}:\ket{\psi}\in\mathcal{H},i\in[1,m]\}$. One important fact about minimal Naimark dilation is that it is unique up to unitary.

\begin{theorem}[\cite{Beneduci_2020}, Theorem 2.22]
\label{thm:minimalsingle}
    Let $(\{P_i\}_{i=1}^{m},V)$, $(\{P'_i\}_{i=1}^{m},V')$ be two minimal Naimark dilations of $\{R_i\}_{i=1}^{m}$. Then there exists unitary $U$ such that $V'=UV$ and $UP_iU^*=P'_i$.
\end{theorem}

We generalise the concept of minimal Naimark dilation in the context of nonlocal strategies.

\begin{definition}
    Let $\{R_{ij}\}_{j=1}^{m_i}$, $1\le i\le n$ be a family of POVMs. A Naimark dilation $(\{P_{ij}\}_{j=1}^{m_i},V)$ of $\{R_{ij}\}$ is \emph{minimal} if, for at least one $i_0\in[1,n]$, $(\{P_{i_0j}\}_j,V)$ is a minimal Naimark dilation of $\{R_{i_0j}\}_j$. 

    Let $S=(\ket{\psi},\{A_{sa}\},\{B_{tb}\})$ be a pure strategy. A pure PVM strategy $S'=(V_A\otimes V_B\ket{\psi},\{P_{sa}\},\{Q_{tb}\})$ is a \emph{minimal} Naimark dilation of $S$, if $(\{P_{sa}\},V_A)$ is a minimal Naimark dilation of $\{A_{sa}\}$, and $(\{Q_{tb}\},V_B)$ is a minimal Naimark dilation of $\{B_{tb}\}$. 
\end{definition}

Minimal Naimark dilation of nonlocal strategies always exists, but is not unique (up to local unitary) in general, since those PVM which are non-minimal could be very different outside the support of the state. Nevertheless, we can show that in a special case, the minimal Naimark dilations of $S$ are equivalent up to local dilation.

\begin{lemma}
    Let $\{R_{ij}\}$ be a family of POVMs on $\mathcal{H}$ with at most one non-projective measurement. Then for any two minimal Naimark dilations $(\{P_{ij}\},V)$, $(\{P'_{ij}\},V')$, there exist unitary $U$ such that
    \begin{align*}
        &UV=V'\\
        &UP_{ij}V\ket{\psi}=P_{ij}V'\ket{\psi},\forall\ket{\psi}\in\mathcal{H}.
    \end{align*}
    \label{lem:uniqueminimal}
\end{lemma}

\begin{proof}
    The case where all $\{R_{ij}\}$ are projections is trivial, because $\{R_{ij}\}$ is the minimal Naimark dilation of itself. Without loss of generality, we assume $\{R_{1j}\}_j$ to be the non-projective measurement. By definition, $\{P_{1j}\}$ and $\{P'_{1j}\}$ are two minimal Naimark dilations of $\{R_{1j}\}$. So, by Theorem \ref{thm:minimalsingle} there exist unitary $U$ such that $UV=V'$, and $UP_{1j}U^*=P'_{1j}$. Also note that $R_{ij}^2=R_{ij}$ for all $i\neq 1$, so
    \begin{align*}
        &\|[VV^*,P_{ij}]V\ket{\psi}\|^2\\
        =&\braket{\psi|V^*P_{ij}V|\psi}-\braket{\psi|V^*P_{ij}VV^*P_{ij}V|\psi}\\
        =&\braket{\psi|(R_{ij}-R_{ij}^2)|\psi}=0.
    \end{align*}
    So $P_{ij}V\ket{\psi}=P_{ij}VV^*V\ket{\psi}=VV^*P_{ij}V\ket{\psi}=VR_{ij}\ket{\psi}$. Similarly, $P'_{ij}V'\ket{\psi}=V'R_{ij}\ket{\psi}$.

    Then the following holds:
    \begin{align*}
        UP_{1j}V\ket{\psi}=&UP_{1j}U^*UV\ket{\psi}=P'_{1j}V'\ket{\psi}\\
        UP_{ij}V\ket{\psi}=&UVR_{ij}\ket{\psi}\\
        =&V'R_{ij}\ket{\psi}\\
        =&P'_{ij}V'\ket{\psi},~\forall i\neq1,
    \end{align*}
    as required.
\end{proof}

For the case of single POVM, any Naimark dilation of $\{R_i\}$ is a Naimark dilation of some minimal Naimark dilation of $\{R_i\}$. It is not true in the case of multiple POVMs or for non-local strategies. Nevertheless, we prove the following:

\begin{lemma}
    \label{lem:dilationofminimal}
    Let $\{R_{ij}\in\mathcal{B}(\mathcal{H})\}_{j=1}^{m_i}$, $1 \leq i \leq n$, be a family of POVMs with at most one non-projective measurement, and let $(\{P_{ij}\in\mathcal{B}(\mathcal{H}')\},V)$ be a Naimark dilation of $\{R_{ij}\}$. Then there exists a minimal Naimark dilation $(\{P^{\min}_{ij}\in\mathcal{B}(\mathcal{H}^{\min})\},V^{\min})$ of $\{R_{ij}\}$ and an isometry $V':\mathcal{H}^{\min}\to\mathcal{H}'$ such that
    \begin{align*}
        &V'V^{\min}=V,\\
        &V'P^{\min}_{ij}V^{\min}\ket{\psi}=P_{ij}V\ket{\psi},\forall\ket{\psi}\in\mathcal{H}.
    \end{align*}
\end{lemma}

\begin{proof}
    The case where all $\{R_{ij}\}$ are projections is trivial. We assume $\{R_{1j}\}_j$ to be the non-projective measurement. Consider the subspace
    $$
    \mathcal{H}^{\min}:=\bigoplus_{j\in[1,m_1]}\mathcal{H}^{\min}_j~\text{of }\mathcal{H}'~,~\text{where}~\mathcal{H}^{\min}_j:=\Span\{P_{1j}V\ket{\psi}:\ket{\psi}\in\mathcal{H}\}.
    $$
    Here $\bigoplus$ refers to the internal direct sum. It is clear that $V\mathcal{H}\subseteq\mathcal{H}^{\min}\subseteq\mathcal{H}'$. Let $V'^{\min}$ be the canonical embedding from $V\mathcal{H}$ to $\mathcal{H}^{\min}$, and $V'$ be the canonical embedding from $\mathcal{H}^{\min}$ to $\mathcal{H}'$. Let $U$ be the unitary from $\mathcal{H}$ to $V\mathcal{H}$. Let $V^{\min}:=V'^{\min}U$. 
    
    We construct 
    $$P^{\min}_{1j}:=(V')^*P_{1j}V',$$ 
    $$P^{\min}_{i1}:=V^{\min}R_{i1}(V^{\min})^*+(I-V^{\min}(V^{\min})^*), i\neq 1$$
    $$
    P^{\min}_{ij}:=V^{\min}R_{ij}(V^{\min})^*, i\neq1,j\neq1.
    $$
    It is clear that $P^{\min}_{ij}$ are projections for $i\neq1$. For $P^{\min}_{1j}$, note that $\mathcal{H}^{\min}_j\subseteq\operatorname{Range}(P_{1j})$, so $P_{1j}$ commutes with $(V')^*V'$. Then $(P^{\min}_{1j})^2=P^{\min}_{1j}$. Also note that $\mathcal{H}^{\min}=\Span\{P_{1j}V\ket{\psi}:\ket{\psi}\in\mathcal{H},j\in[1,m_1]\}$, so $\{P^{\min}_{ij}\}$ is a minimal Naimark dilation of $\{R_{ij}\}$. The following holds:
    \begin{align*}
        V'P^{\min}_{1j}V^{\min}\ket{\psi}=&V'(V')^*P_{1j}V'V^{\min}\ket{\psi}\\
        =&P_{1j}V'(V')^*V'V^{\min}\ket{\psi}=P_{1j}V\ket{\psi},\\
        V'P^{\min}_{ij}V^{\min}\ket{\psi}=&V'V^{\min}R_{ij}\ket{\psi}=VR_{ij}V^*V\ket{\psi}=P_{ij}V\ket{\psi},\forall i\neq1.
    \end{align*}
    So we conclude that $(\{P^{\min}_{ij}\},V^{\min})$ satisfies the required property.
\end{proof}

Applying Lemma \ref{lem:uniqueminimal} and \ref{lem:dilationofminimal} in the context of non-local strategies, we have the following:

\begin{proposition}
    Let $\tilde S$ be a pure full-rank strategy with at most one non-projective measurement on each side. Then any Naimark dilations of $\tilde S$ are local-dilations of each other.
    \label{prop:Naimarkequal}
\end{proposition}

\begin{proof}
    Consider two Naimark dilations $S_1$ and $S_2$ of $\tilde S$. By Lemma \ref{lem:dilationofminimal}, there exists minimal Naimark dilations $\tilde{S}^{\min}_1$ and $\tilde{S}^{\min}_2$ of $\tilde S$ such that $\tilde{S}^{\min}_1\xhookrightarrow{}S_{1}$, $\tilde{S}^{\min}_2\xhookrightarrow{}S_{2}$. Then from Proposition \ref{prop:otherdirection}
    $$
    S_{1}\xhookrightarrow{}\tilde{S}^{\min}_1,S_2\xhookrightarrow{}\tilde{S}^{\min}_2.
    $$

    Also, from Lemma \ref{lem:uniqueminimal} we know that $\tilde{S}^{\min}_1$ and $\tilde{S}^{\min}_2$ are local dilations of each other. So we conclude that $S_{1}\xhookrightarrow{}S_2$ and $S_{2}\xhookrightarrow{}S_1$.
\end{proof}

\begin{theorem}
    Let $\tilde S$ be a pure full-rank strategy with at most one non-projective measurement on each side. Then if $p$ (or $G$) full-rank self-tests $\tilde S$, $p$ (or $G$) also PVM self-tests any Naimark dilation of $\tilde S$.
\end{theorem}

\begin{proof}
    Consider a pure PVM strategy $S_{\operatorname{PVM}}$ that generates the same correlation as $\tilde S=(\ket{\tilde\psi},\{\tilde A_{sa}\},\{\tilde B_{tb}\})$. From full-rank self-test, the restriction of $S_{\operatorname{PVM}}$ is equivalent to $\tilde S$ attached with some auxiliary state up to local unitary. In other words, $S_{\operatorname{PVM}}$ is a Naimark dilation of $\tilde S\otimes\ket{\aux}=(\ket{\tilde\psi}\ket{\aux},\{\tilde A_{sa}\otimes\mathbbm{1}_{\aux,A}\},\{\tilde B_{tb}\otimes\mathbbm{1}_{\aux,B}\})$. Note that $\ket{\tilde\psi}\ket{\aux}$ is also full-rank, then from Proposition \ref{prop:Naimarkequal} and the transitivity of local dilation, $$S_{\operatorname{PVM}}\xhookrightarrow{}\tilde{S}_{\operatorname{Naimark}}\otimes\ket{\aux}\xhookrightarrow{}\tilde{S}_{\operatorname{Naimark}}
    $$
    for any Naimark dilation $\tilde{S}_{\operatorname{Naimark}}$ of $\tilde S$.
\end{proof}

\begin{corollary}
\label{cor:PVMselftestnoselftest}
    Let $\tilde p$ be the correlation generated by pure non-projective strategy $\tilde S=(\ket{\Phi^+},\{\mathcal{X},\mathcal{Z}\},\{\mathcal{H},\mathcal{G},\mathcal{M}\})$. Then $\tilde p$ is a PVM self-test for any Naimark dilation of $\tilde S$, but not a pure self-test. 
\end{corollary}

Now we present a minimal Naimark dilation for $\tilde{S}$. Since the measurements for Alice are projective, they are minimal themselves. For Bob, let $V$ be the canonical embedding $\mathbb{C}^2\to\mathbb{C}^3$ (that is, in the computational basis $V=\mathbbm{1}_{3\times2}$). Then for $\mathcal{M}$, let rank-1 projections $M'_i=\proj{e_i}$ for $i=0,1,2$, where
\begin{align}
\ket{e_{0}} &= \frac{1}{\sqrt{3}}
\left( \begin{array}{c}
\sqrt{2} \\
0 \\
1
\end{array} \right),\\
\ket{e_{1}} &= \frac{1}{\sqrt{6}}
\left( \begin{array}{c}
-1 \\
-\sqrt{3} \\
\sqrt{2}
\end{array} \right),\\
\ket{e_{2}} &= \frac{1}{\sqrt{6}}
\left( \begin{array}{c}
- 1 \\
\sqrt{3} \\
\sqrt{2}
\end{array} \right).
\end{align}

And let $\mathcal{M}'=\{M'_0,M'_1,M'_2\}$. According to \cite{Beneduci_2020}, $(\mathcal{M}',V)$ is a minimal Naimark dilation of $\mathcal{M}$. For $\mathcal{G}$ and $\mathcal{H}$, since they are projective themselves, we just need to ensure their projectiveness outside the range of $V$ when we extend them. To do this, we let $H_{\pm}$, $G_{\pm}$ be the $\pm1$-eigenspace projection of $H,G$, respectively. Define $H'_{+}=VH_{+}V^*+\mathbbm{1}-VV^*$, $H'_{-}=VH_{-}V^*$ (that is, $H'_{+}=H_+\oplus\mathbbm{1},H'_{-}=H_-\oplus0$), and $G'_{+}=VG_{+}V^*+\mathbbm{1}-VV^*$, $G'_{-}=VG_{-}V^*$. Then $(\mathcal{H}'=\{H'_{+},H'_{-}\},V)$, $(\mathcal{G}'=\{G'_{+},G'_{-}\},V)$ are Naimark dilations if $\mathcal{H}$ and $\mathcal{G}$, respectively. So we conclude that $\tilde S_{PVM}=(\mathbbm{1\otimes}V\ket{\Phi^+},\{\mathcal{X},\mathcal{Z}\},\{\mathcal{H}',\mathcal{G}',\mathcal{M}'\})$ is a minimal Naimark dilation of $\tilde S$.

\subsection{Separating (standard) self-tests and abstract state self-tests}

We show that Corollary \ref{cor:PVMselftestnoselftest} also answers the open question raised in \cite{paddock2023operatoralgebraic}, separating abstract state self-testing defined therein and (standard) self-testing in a case where there is no full-rank strategy in a certain class of strategies (namely, the class of all pure PVM strategies). Recall that, in an abstract state self-test the higher order moments are the same for all strategy inducing the correlation $\tilde p$.

\begin{definition}[\cite{paddock2023operatoralgebraic}]
    Let $t\subseteq {\{\text{pure},\text{full-rank},\text{PVM}\}}$. A correlation $\tilde p$ is an abstract state $t$ self-test if for every $k,l\ge1$, $a_1,\dots,a_k\in\mathcal{A}$,$s_1,\dots,s_k\in\mathcal{S}$,$b_1,\dots,b_l\in\mathcal{B}$,$t_1,\dots,t_l\in\mathcal{T}$, the value
    $$
    \braket{\psi|A_{s_1a_1}\cdots A_{s_ka_k}\otimes B_{t_1b_1}\cdots B_{t_lb_l}|\psi}
    $$
    is the same across all $t$ strategies inducing the correlation $\tilde p$.
\end{definition}

\begin{proposition}
    Let $\tilde p$ be the correlation generated by the pure non-projective strategy $\tilde S=(\ket{\Phi^+},\{\mathcal{X},\mathcal{Z}\},\{\mathcal{H},\mathcal{G},\mathcal{M}\})$. Then $\tilde p$ is not an abstract state PVM self-test.
\end{proposition}

\begin{proof}
    According to the definition of abstract state self-testing, it suffices to find two pure PVM strategies for $p$ which give different higher-order moments.

    Define $S^1_{PVM}=(\mathbbm{1\otimes}V\ket{\Phi^+},\{\mathcal{X},\mathcal{Z}\},\{\mathcal{H}',\mathcal{G}',\mathcal{M}'\})$ as in the previous subsection. Now consider another dilation $\mathcal{H}''$ of $\mathcal{H}$, namely, $H''_{+}=VH_{+}V^*$, $H''_{-}=VH_{-}V^*+\mathbbm{1}-VV^*$ (that is, $H''_{+}=H_+\oplus0,H''_{-}=H_-\oplus\mathbbm{1}$). Let $S^2_{PVM}=(\mathbbm{1\otimes}V\ket{\Phi^+},\{\mathcal{X},\mathcal{Z}\},\{\mathcal{H}'',\mathcal{G}',\mathcal{M}'\})$. Then direct calculation shows that
    \begin{align*}
        &\braket{\Phi^+|(\mathbbm{1}\otimes V^*)(\mathbbm{1}\otimes(M_0'H'_{+}M_0'))(\mathbbm{1}\otimes V)|\Phi^+}=\frac{4-\sqrt{2}}{18},\\
        &\braket{\Phi^+|(\mathbbm{1}\otimes V^*)(\mathbbm{1}\otimes(M_0'H''_{+}M_0'))(\mathbbm{1}\otimes V)|\Phi^+}=\frac{2-\sqrt{2}}{18}.
    \end{align*}
    So $S^1_{\text{PVM}}$ and $S^2_{\text{PVM}}$ are of different higher order moments.
\end{proof}

Note that by to \cite[Theorem 3.5]{paddock2023operatoralgebraic}, abstract state self-testing is equivalent to (standard) self-testing under the condition that $\tilde p$ is extreme and there exists a full-rank $t$ strategies inducing the correlation $\tilde p$. Therefore, our results indicates that the condition of \cite[Theorem 3.5]{paddock2023operatoralgebraic} is crucial: there exists extreme correlation $\tilde p$ such that, the class of PVM strategies admits no full-rank strategy for $\tilde p$, where $\tilde p$ is a (standard) PVM-self-test but not an abstract state PVM-self-test.
    \section{Acknowledgements}

This work is funded by the European Union under the Grant Agreement No 101078107, QInteract and Grant Agreement No 101017733, VERIqTAS as well as VILLUM FONDEN via the QMATH Centre of Excellence (Grant No 10059) and Villum Young Investigator grant (No 37532). S. S. was funded by the Deutsche Forschungsgemeinschaft (DFG, German Research Foundation) under Germany's Excellence Strategy - EXC 2092 CASA - 390781972. He furthermore has received funding from the European Union's Horizon 2020 research and innovation programme under the Marie Sklodowska-Curie grant agreement No. 101030346. P. B. acknowledges the support from CNPq. J. K. is supported by the HOMING grant from the Foundation for Polish Science. We thank Jurij Vol\v{c}i\v{c} for valuable discussion on Naimark dilation.

	\bibliographystyle{halpha}
    \bibliography{biblo}
    \begin{appendix}
    \section{Self-testing from correlation}\label{app:probdist}
In the main paper, we have mostly focused on self-testing from the perspective of non-local games. A different way to define self-testing is from the perspective of probability distributions, or correlations. In this setting, instead of having a game self-test a strategy $\tilde S$, we have a correlation $p$ that self-tests a strategy $\tilde S$. We will denote the correlation generated by the strategy $S$ as $p_S$. 

We say that the strategy $S = (\ket{\psi},\{A_{sa}\},\{B_{tb}\})$ generates the correlation $p$ if for all $a,b,s,t$

$$p(a,b|s,t) = \braket{\psi|A_{sa}\otimes B_{tb}|\psi}.$$

This is the definition of self-testing that is used in \cite{SB}, though again augmented with the same qualifiers as in \pref{def:selftest}.

\begin{definition}[Self-testing from correlation]\label{def:prob_selftest}
    Let $t\subseteq {\{\text{pure},\text{full-rank},\text{PVM}\}}$ and let $\tilde{S}$ be a pure strategy. We say that a correlation $p_{\tilde{S}}$ \emph{$t$-\emph{self-tests}} a (canonical) strategy $\tilde{S}$ if $S \xhookrightarrow{} \tilde{S}$ for every $t$-strategy $S$ where $p_S = p_{\tilde{S}}$.
\end{definition}

In self-testing from correlation, our result for lifting the full-rank or PVM assumption (Theorem \ref{thm:full_to_any_rank_rob} and \ref{thm:PVMtoPOVMrob}) still hold. However, our results for lifting the purity assumption (Theorem \ref{thm:pure_to_mix_rob}) does not translate directly into this setting. Here, we show how to translate the proof of Theorem \ref{thm:pure_to_mix_rob} for self-testing from correlation, by spotting significant points where they differ. In this appendix, we only consider exact self-testing.

The first one is in the proof of Lemma \ref{lem3robmixed}, where we have no game operator for correlation. Nevertheless, we can show the following analog of Proposition \ref{prop:dim1}:

\begin{proposition}
    Let $\tilde S = (\ket{\tilde \psi},\{\tilde{A}_{sa}\},\{\tilde{B}_{tb}\})$ be the canonical full-rank strategy $t$-self-tested by correlation $p_{\tilde S}$. Then any state $\ket\psi \in \mathcal{H}_{\tilde{A}} \otimes \mathcal{H}_{\tilde{B}}$ satisfying
$$ \braket{\psi|\tilde{A}_{sa}\otimes\tilde{B}_{tb}|\psi}=p_{\tilde{S}}(a,b|s,t)$$
for all $a,s,t,b$ must be that $\proj{\psi}=\proj{\tilde\psi}$.
\label{prop:1dim_correlation}
\end{proposition}

\begin{proof}
    Let $S=(\ket{\psi},\{\tilde{A}_{sa}\},\{\tilde{B}_{tb}\})$. By $t$-self-testing, $S \xhookrightarrow{V_A\otimes V_B,\ket{\aux}} \tilde S$, which implies 
    \begin{equation}\label{eq:prob_dist_uni_isom}
        (V_A\otimes V_B)(A_{sa}\otimes B_{tb})\ket{\psi} = (A_{sa}\otimes B_{tb})\ket{\tilde\psi}_{AB}\otimes \ket{\aux}_{\hat A \hat B}.
    \end{equation}
    By \pref{lem:minselftest}, $\ket{\psi}$ has at least as large Schmidt rank as $\ket{\tilde \psi}$. The fact that they live on the same space and $\ket{\tilde \psi}$ has full Schmidt rank implies $\ket{\psi}$ also has full Schmidt rank. This means that $\ket{\aux}$ has Schmidt rank $1$, and so is a product state, $\ket{\aux}_{\hat A \hat B} = \ket{\aux}_{\hat A}\otimes\ket{\aux}_{\hat B}$. 

    Consider the Schmidt decomposition of $\ket{\psi} = \sum_i \lambda_i \ket{\alpha_i}\ket{\beta_i}$. If we now trace out $\mathcal{H}_{\tilde B}\otimes \mathcal{H}_{\hat B}$ from \pref{eq:prob_dist_uni_isom}, and sum over $a,b$ we get
    \begin{equation}
        \sum_i \lambda_i V_A \proj{\alpha_i} V_A^* 
        = \tr_{\tilde B}(\proj{\tilde\psi}_{\tilde A \tilde B}) \otimes \proj{\aux}_{\hat A}.
    \end{equation}
    By \pref{lem:sum_to_pure} we can conclude that
    $$V_A \proj{\alpha_i} V_A^* = \proj{\phi_i}_{\tilde A} \otimes \proj{\aux}_{\hat A}$$
    for some state $\ket{\phi_i}\in \mathcal{H}_{\tilde A}$ for all $i$. This implies that $V_A \ket{\alpha_i} = \ket{\phi_i}_{\tilde A} \ket{\aux}_{\hat A}$ where we have absorbed a potential global phase into $\ket{\phi_i}$. Since $\ket{\psi}$ has full Schmidt rank, $\Span(\{\alpha_i\}_i) = \mathcal{H}_A$. This implies for all $\ket{v}\in \mathcal{H}_A$, $ V_A \ket{v} = \ket{v'}\ket{\aux}_{\hat A}$. This directly implies that $(\mathbbm{1}_{\tilde A}\otimes \bra{\aux}_{\hat A})V_A$ is unitary. A similar argument for $\mathcal{H}_B$ shows that $(\mathbbm{1}_{\tilde B}\otimes \bra{\aux}_{\hat B})V_B$ is unitary. So we conclude that exists unitaries $U_A:=(\mathbbm{1}_{\tilde A}\otimes \bra{\aux}_{\hat A})V_A$, $U_B:=(\mathbbm{1}_{\tilde B}\otimes \bra{\aux}_{\hat B})V_B$ such that for all $s,t,a,b$
    $$U_A\otimes U_B(A_{sa}\otimes B_{tb})\ket{\psi} = (A_{sa}\otimes B_{tb})\ket{\tilde\psi}.$$
    
    Since $\ket{\tilde \psi}$ has full Schmidt rank, $\Supp_{\tilde A}(\ket{\psi}) = \mathcal{H}_{\tilde A}$. This implies that for all $\ket{v}\in\mathcal{H}_{\tilde A}$, $U_A \tilde A_{sa} U_A^*\ket{v} = \tilde A_{sa}\ket{v}$ for all $s,a$. This directly gives that $U_A \tilde A_{sa} U_A^* = \tilde A_{sa}$, and so $[\tilde A_{sa},U_A] = 0$ for all $s,a$. 
    By \pref{lem:minselftest} the state $\ket{\tilde \psi}$ has minimum Schmidt rank across all states that can give rise to $p_{\tilde S}$ using some local measurements. Hence, strategy $\tilde S$ has the minimum (local) dimension among all strategies that give rise to $p$. It now follows that the matrix algebra generated by all the $\tilde{A}_{sa}$ is irreducible and thus $\langle \tilde A_{sa} \rangle_{sa} = \mathbb{M}_{d\times d}$, where $d=\dim(\mathcal H_{\tilde A})$.
    The only matrix that commutes with all $d\times d$ matrices is $U_A = c\mathbbm{1}$. After applying a similar argument to Bob, we obtain that $(U_A\otimes U_B)\ket{\psi} = \ket{\tilde \psi}$, where both $U_A$ and $U_B$ are proportional to identity. This establishes the desired statement.
\end{proof}

The second major difference is in the proof of Lemma \ref{lem3robmixed}, where we decomposed $X\ket{\psi}_{ABP}$ in the eigenspace of $\tilde W$. In \pref{def:selftest}, we require the strategies to be optimal, which gives us an extremality conditions on the game value. Here, in the correlation case, we would instead have the correlation be extreme in the set of quantum correlations.

\begin{lemma}
    \label{lem:extremalityused}
    Let $\tilde S$ be a pure full-rank strategy self-tested by a extreme correlation $p_{\tilde S}$. Then for any mixed state $\rho\in\mathcal{B}(\mathcal{H}_{\tilde A}\otimes\mathcal{H}_{\tilde B})$ satisfying
    $$
    \tr[\rho\tilde A_{sa}\otimes\tilde B_{tb}]=p_{\tilde S}(a,b|s,t), \forall a,b,s,t
    $$
    must be so that $\rho=\proj{\tilde{\psi}}$.
\end{lemma}
\begin{proof}
    Consider the spectral decomposition of $\rho$
    $$
    \rho=\sum_i\lambda_i\proj{\phi_i},\ket{\phi_i}\in\mathcal{H}_{\tilde A}\otimes\mathcal{H}_{\tilde B}.
    $$
    Then 
    $$
    p_{\tilde S}(a,b|s,t)=\tr[\rho\tilde A_{sa}\otimes\tilde B_{tb}]=\sum_i \lambda_i\braket{\phi_i|\tilde A_{sa}\otimes\tilde B_{tb}|\phi_i}.
    $$
    From the extremality of $p_{\tilde S}$, each $\ket{\phi_i}$ must have the same correlation using the measurements $\tilde A_{sa}\otimes\tilde B_{tb}$. By Proposition \ref{prop:1dim_correlation}, $\proj{\phi_i}=\proj{\tilde\psi}$. So $\rho=\proj{\tilde{\psi}}$.
\end{proof}

Now we can state the primary result of this appendix. We are not going to fully prove this, since the proof is essentially the same as \pref{thm:pure_to_mix_rob}. We are instead going to state the places where they differ.

\begin{theorem}\label{thm:pure_to_mix_correlation}
	Let $t\in\{\text{PVM,POVM}\}$, and correlation $p$ pure $t$ self-tests $\tilde{S}$, where $\ket{\tilde{\psi}}$ has full Schmidt rank. If $p_{\tilde{S}}$ is extreme in the set of quantum correlations, then $p$ mixed $t$ self-tests $\tilde S$.
\end{theorem}
\begin{proof}[Sketch]
    The proof is very similar to the one in \pref{thm:pure_to_mix_rob}. Following the proof of Lemma \ref{lem1robmixed}, we can show that there exists local isometry $X$ such that $(X\ket{\psi},\{\tilde A_{sa}\otimes\mathbbm{1}_{\check{A}}\},\{\tilde B_{tb}\otimes\mathbbm{1}_{\hat{B}}\otimes\mathbbm{1}_{P})$ is a local dilation of the purification of an arbitrary strategy. Let
    \begin{equation*}
		\rho_{\tilde A\tilde B}' := \tr_{\check{A}\hat{B}P}[X\proj{\psi}X^*] = \sum_{i = 0}^{\operatorname{rank} \rho_{\tilde A\tilde B}'-1} p_i\proj{\psi_i'}_{\tilde A\tilde B}.
	\end{equation*}
    
    Observe that by the extremality of $p$, each of the $\ket{\psi_i'}_{\tilde A\tilde B}$ must have the same correlation using the same measurements $\tilde A_{sa}\otimes\tilde B_{tb}$. By Lemma \ref{lem:extremalityused}, $X\ket{\psi}=\ket{\tilde\psi}\ket{\aux}$ for some auxiliary state $\ket{\aux}\in\mathcal{H}_{\check{A}\hat{B}P}$.

    This proves that $p_{\tilde S}$ mixed $t$-self-tests $\tilde S$.
\end{proof}
    \end{appendix}
\end{document}